%% file: main.tex
\documentclass[twoside,11pt]{article}

\usepackage{blindtext}

\usepackage{jmlr2e}
\usepackage{comment}
\usepackage{pifont}
\usepackage{amsmath}
\usepackage{tikz}
\usetikzlibrary{arrows.meta, positioning, fit, backgrounds, calc, decorations.pathreplacing, shapes.geometric,patterns}
\usepackage{pgfplots}
\usepackage{subcaption}
\pgfplotsset{compat=1.18} % Ensure compatibility with pgfplots
\usepackage{algorithm}
\usepackage{algorithmic}
\usepackage{booktabs}
\usepackage{cleveref}
\usepackage[colorinlistoftodos]{todonotes}

\firstpageno{1}
\input{macros}

\usepackage{lastpage}
\usepackage[T1]{fontenc}
\definecolor{inputblue}{HTML}{4A90D9}
\definecolor{boundgreen}{HTML}{5BA870}
\definecolor{milporange}{HTML}{E8913A}
\definecolor{optpurple}{HTML}{8B6DB5}
\definecolor{reachred}{HTML}{D9534F}
\definecolor{safeteal}{HTML}{2A9D8F}
\definecolor{lightgray}{HTML}{F5F5F5}
\definecolor{bordergray}{HTML}{CCCCCC}

\ShortHeadings{Polyedral Enclosures}{Akinwande et. al.}

\newcommand{\modded}[1]{{}}
\begin{document}

\title{Polyhedral Enclosures: An Efficient Combinatorial Abstraction for Nonlinear Neural Feedback Systems}

\author{\name I. Samuel Akinwande \email samakin@stanford.edu \\
       \addr Department of Aerospace, Aeronautics and Astronautics\\
       Stanford University\\
       Stanford, CA 94305, USA
       \AND
       \name Chlesea Sidrane \email chelse@kth.se \\
       \addr  Division of Robotics, Perception and Learning \\%, Intelligent Systems Department, School of Electrical Engineering \& Computer Science\\
       KTH Royal Institute of Technology\\
       Stockholm, Sweden
       \AND
       \name Mykel J. Kochenderfer \email mykel@stanford.edu \\
         \addr Department of Aerospace, Aeronautics and Astronautics\\
       Stanford University\\
       Stanford, CA 94305, USA
       \AND
       \name Clark Barrett \email barrettc@stanford.edu \\
       \addr Department of Computer Science\\
        Stanford University\\
       Stanford, CA 94305, USA}

\editor{My editor}

\maketitle

\begin{abstract}%   <- trailing '%' for backward compatibility of .sty file
As dynamical systems controlled by neural networks become increasingly prevalent, it is critical to ensure their safe operation.  
Although efficient techniques exist to handle neural systems with \emph{linear} transition functions, few scalable methods address the
\emph{nonlinear} case. 
We propose a novel algorithm for verifying nonlinear neural feedback systems using forward reachability analysis. Our algorithm leverages the structure of the nonlinear transition functions to compute tight linear abstractions which we call polyhedral enclosures.  These are then encoded as mixed-integer linear programs (MILPs) and solved to yield a sound over-approximation of the forward-reachable set. We evaluate our algorithm on representative benchmarks and demonstrate significant improvements over the previous state of the art.
\end{abstract}


\section{Introduction}
\label{sec:intro}
\input{sections/intro}

\section{Background}
\label{sec:Back}
\modded{Pipeline figure and NFS figure}
We rely on concepts from the theory of polyhedra. The following sections introduce notation, and help describe the necessary background. 
\begin{figure}
    \centering
    \resizebox{0.9\textwidth}{!}{\input{figures/pipeline}}
    \caption{\footnotesize The OvertPoly algorithmic pipeline: Given a neural feedback system $\mathcal{D}$, we compute polyhedral enclosures ($B_t$) for its nonlinear dynamics, encode them (along with the network) as mixed integer linear programs ($M_0,\ldots,M_n$), and solve the MILPs to compute reachable sets}
    \label{fig:overview}
\end{figure}
\subsection{Notation}
\label{ssec:notation}
\input{sections/notation}

\subsection{Polyhedra}
\label{ssec:polyhedra}
\input{sections/polyhedra}

\subsection{Bounding Convex Functions}
\label{sec:overt}
\input{sections/overt}

\subsection{ReLU Functions as Mixed Integer Constraints}
\label{sec:MipVerify}
\input{sections/mipVerify}

\section{Neural Feedback Systems}
\begin{figure}
    \centering
    \resizebox{0.4\textwidth}{!}{\input{figures/nfs}}
    \caption{Closed loop depiction of a neural feedback system}
    \label{fig:nfs}
\end{figure}
\label{sec:Prob}
\input{sections/problem}

\section{Polyhedral Enclosures}
\modded{Introduced illustrating steps of the theory}
\label{sec:PolyEnc}
\input{sections/theory}

\section{The OvertPoly Algorithm}
\label{sec:alg}
\input{sections/algorithm}

\section{Optimizations}
\modded{Added this section to discuss the optimizations we introduced. 6.3 onwards is entirely new}
\label{sec:opt}
In addition to the dependency graph structure, we implement a few additional optimizations. We elaborate on the dependency graph structure and discuss the other optimizations below.
\subsection{Using dependency graphs}
\input{sections/optimizations/depgraph}
\subsection{Tightening nonlinear composition}
\label{ssec:mccormick}
\input{sections/optimizations/mccormick}
\subsection{Tightening pre-activation bounds}
\label{ssec:crown}
\input{sections/optimizations/crown}
\subsection{Improving LP relaxations for ReLU networks}
\label{ssec:anderson}
\input{sections/optimizations/ideal_relu}
\subsection{Obtaining a more compact encoding for enclosures}
\label{ssec:dcc}
\input{sections/optimizations/dcc}

\section{Implementation}
\label{sec:imp}
\input{sections/implementation}

% \section{Optimizations}
% \label{sec:opt}
% \subsection{Tighter pre-activation bounds}
% \subsection{Disaggregated convex combination encoding}
\section{Evaluation}
\modded{Revamped evals section to contextualize new experiments, more thoroughly discuss results, and discuss ablations/optimizations}
\label{sec:eval}
\input{sections/eval}

\section{Conclusions}
\label{sec:conclusion}
\input{sections/endSections}

\newpage
\bibliography{main}

\end{document}

%% file: macros.tex
\newcommand{\dynsys}{\ensuremath{\mathcal{D}}\xspace}
\newcommand{\ints}{\ensuremath{\mathbb{Z}}\xspace}
\newcommand{\nats}{\ensuremath{\mathbb{N}}\xspace}
\newcommand{\reals}{\ensuremath{\mathbb{R}}\xspace}

\newcommand{\ReLU}{\ensuremath{\mathsf{ReLU}}\xspace}
\newcommand{\traj}{\ensuremath{\tau}\xspace}
\newcommand{\nextF}{\ensuremath{\mathit{next}}\xspace}
\newcommand{\noise}{E\xspace}

\newcommand{\init}{I\xspace}
\newcommand{\trans}{\mathbf{F}\xspace}
\newcommand{\ctrl}{u\xspace}
\newcommand{\timestep}{\ensuremath{\delta}\xspace}
\newcommand{\horizon}{\ensuremath{T}\xspace}
\newcommand{\reach}{G\xspace}
\newcommand{\avoid}{A\xspace}

\newcommand{\node}{\ensuremath{\mathcal{N}}\xspace}

\newcommand{\poly}{\ensuremath{\mathcal{P}}\xspace}
\newcommand{\simp}{\ensuremath{\mathcal{S}}\xspace}
\newcommand{\svert}{\ensuremath{\mathbf{vert}}\xspace}

\newcommand{\polyenc}{\ensuremath{\mathcal{E}}\xspace}
\newcommand{\vset}{\ensuremath{\mathcal{B}}\xspace}
\newcommand{\pset}{\ensuremath{P}\xspace}

\newcommand{\stateset}{\ensuremath{\mathcal{X}}\xspace}
\newcommand{\State}{\stateset}
\newcommand{\nextState}{\ensuremath{\mathcal{X}_{t+1}}\xspace}
\newcommand{\hnextState}{\ensuremath{\hat{\mathcal{X}}_{t+1}}\xspace}
\newcommand{\model}{\ensuremath{\mathcal{M}}\xspace}
\newcommand{\scomp}{\ensuremath{\Delta}\xspace}

\newcommand{\cell}{\ensuremath{\mathcal{C}}\xspace}
\newcommand{\conv}{\ensuremath{\mathbf{conv}}\xspace}
\newcommand{\dom}{\ensuremath{\mathit{dom}}\xspace}

\newcommand{\powset}[1]{\ensuremath{2^{#1}}\xspace}
\newcommand{\uni}{\ensuremath{\mathcal{U}}\xspace}

\newcommand{\overt}{\ensuremath{\mathit{Overt}}\xspace}
\newcommand{\compose}{\ensuremath{\mathit{Compose}}\xspace}
\newcommand{\unit}[1]{\ensuremath{\mathbf{e}_{#1}}\xspace}

\newcommand{\vleq}{\ensuremath{\preccurlyeq}\xspace}
\newcommand{\vgeq}{\ensuremath{\succcurlyeq}\xspace}

\newcommand{\pend}{\ensuremath{\mathcal{P}}\xspace}
\newcommand{\acc}{\ensuremath{\mathcal{A}}\xspace}
\newcommand{\tora}{\ensuremath{\mathcal{T}}\xspace}
\newcommand{\point}{\ensuremath{\vec{q}}\xspace}
\newcommand{\state}{\ensuremath{\vec{x}}\xspace}

\newcommand{\interp}{\ensuremath{+~\!\!\!\!{\text{\tiny$_\Delta$}}\,}\xspace}
\newcommand{\interpi}{\ensuremath{+~\!\!\!\!{\text{\tiny$_{\Delta_i}$}}\,}\xspace}

\newcommand{\lift}[3]{\ensuremath{#1\!\!\uparrow\!\!^{#2}_{#3}}\xspace}
\newcommand{\fg}{\ensuremath{\mathit{fg}}\xspace}

\newcommand{\cells}{\ensuremath{\mathit{Cells}}\xspace}

\newcommand{\pl}{\ensuremath{\vec{p}^{\:l}}}
\newcommand{\pu}{\ensuremath{\vec{p}^{\:u}}}

\newcommand{\bound}{\ensuremath{\mathrm{Bound}}\xspace}

%% file: sections/intro.tex
The success of neural networks in the fields of drone racing~\citep{kaufmann2023champion}, autonomous driving~\citep{ettinger2021large}, and system identification~\citep{duong2024port} illustrates the growing interest in using neural networks as controllers for dynamical systems. We refer to these as \emph{neural feedback systems}, and many of their potential uses are in safety-critical settings. In order to realize this potential, it is crucial to develop verification techniques to ensure their safety.
Although there is a rich body of work on verification for classical control systems~\citep{tomlin2003computational,chen2013flow,bansal2017hamilton}, 
these techniques are often ill-suited for neural feedback systems due to nonlinearities and the common practice of viewing neural networks as black boxes. To address this gap,  verification approaches based on reachability analysis have been developed~\citep{ivanov2020verifying,everett2021reachability,sidrane2022overt,wang2023polar,zhang2023reachability,kochdumper2023constrained}. 
These methods can be broadly classified into two categories: propagation-based methods and combinatorial methods~\citep{everett2021neural}, 
each offering distinct advantages and disadvantages.

Following the terminology of~\citep{everett2021neural}, we define 
\emph{propagation-based methods} as those that systematically propagate an initial set through both the dynamical system and the neural network. These methods often rely on abstractions, such as Taylor models~\citep{dutta2019reachability,huang2022polar}, Bernstein polynomials~\citep{fan2020reachnn}, zonotopes~\citep{SchillingFG22}, and polynomial zonotopes~\citep{kochdumper2023open}, to achieve efficient computations. To address the loss of precision caused by these abstractions, some tools introduce refinement algorithms that can improve the tightness of the analysis~\citep{ladner2023automatic,10156051,9296363}.
CORA is a prominent tool for verifying neural feedback systems using abstraction propagation \citep{althoff2016cora}. 
It computes reachable sets by combining approximations of the system dynamics with non-convex network abstractions (e.g., polynomial zonotopes). 

For the state space, CORA employs conservative linear~\citep{althoff2008reachability} and polynomial~\citep{althoff2013reachability} abstractions. To handle the neural network, it uses reachability analysis~\citep{kochdumper2020sparse,kochdumper2023open} to compute a non-convex enclosure of the network’s outputs. By leveraging these abstraction-based methods alongside non-convex set representations, CORA trades off precision for computational efficiency making it
%balances computational efficiency with a high degree of precision. 
one of the most competitive tools for formally verifying neural feedback systems. Although other propagation-based tools exist, many either lack support for discrete-time nonlinear neural feedback systems~\citep{bogomolov2019juliareach} or incorporate CORA’s core algorithm~\citep{lopez2023nnv}.

In contrast, \emph{combinatorial methods} verify properties by solving combinatorial problems. Existing approaches include techniques that model neural feedback systems as hybrid systems~\citep{ivanov2019verisig} or as hybrid zonotopes~\citep{siefert2023successor,zhang2023reachability}, 
as well as methods that encode the problem as marching trees~\citep{vincent2021reachable} or integer linear programs~\citep{sidrane2022overt}. By emphasizing the combinatorial structure of the neural network controller, these methods trade off computational complexity for precision. 
A notable example is the OVERTVerify algorithm~\citep{sidrane2022overt}, which computes reachable sets for nonlinear neural feedback systems. OVERTVerify uses the OVERT algorithm and relational overapproximations to construct constraints that implicitly bound multivariate nonlinear functions. 

While alternative combinatorial tools exist for ReLU networks with nonlinear dynamics 
(e.g., those employing hybrid zonotopes~\citep{siefert2023successor}), 
OVERTVerify has thus far remained one of the most scalable combinatorial approaches for capturing reachability, due to its ability to efficiently handle larger systems while preserving high precision in its reachability analysis.

Combinatorial methods, including OVERTVerify, are often prohibitively expensive computationally, and thus, many successful verification tools instead rely on abstraction propagation and refinement~\citep{lopez2023arch}.

In this paper, we introduce the OvertPoly algorithm, which aims to address some of the limitations of combinatorial methods. We exploit a structured representation of the neural feedback system to create a precise combinatorial method with computational performance comparable to propagation-based methods. 

\modded{Modified this section to provide a better comparison}
Our bounding approach is based on the OVERT algorithm~\citep{sidrane2022overt}, and is closely related to the notion of \emph{convex envelopes}~\citep{rikun1997convex}. A convex envelope is the tightest convex relaxation of a nonlinear function over a given domain and plays a central role in nonlinear programming. Substantial effort has been devoted to constructing such envelopes. For example, \citet{tawarmalani2013explicit} identify subdivisions of hyperrectangles over which convex envelopes can be constructed for a class of supermodular functions, while \citet{nagarajan2019adaptive} develop an adaptive partitioning algorithm for constructing convex envelopes of polynomial functions.
The univariate polyhedral bounds produced by the OVERT algorithm can be interpreted as instances of convex envelopes. However, our work departs from the classical convex envelope literature in two key ways. First, we focus on the composition of such enclosures, which is essential for analyzing neural feedback systems. Second, our framework applies to a broader class of nonlinear functions, whereas existing convex envelope constructions are often restricted to specific function classes, most notably polynomials.
\begin{comment}
Our work is also related to \emph{subdividable linear efficient function enclosures} (SLEFEs)~\citep{LUTTERKORT2001851,peters2002optimality} and \emph{minimum volume simplicial enclosures} (MINVO)~\citep{tordesillas2022minvo}. SLEFEs enclose multivariate polynomials by computing upper and lower bounds using a second difference operator and its corresponding basis. In contrast, our method encloses multivariate functions by composing enclosures of univariate functions.
Rather than using simplices, as in MINVO enclosures, our approach makes use of general polyhedra to bound functions. Crucially, our method can be applied to a broader class of functions than those handled by either SLEFEs or MINVO enclosures. 
\end{comment}
Our contributions are listed below:
\begin{itemize}
    \item We introduce \emph{polyhedral enclosures}, a novel combinatorial abstraction for multivariate nonlinear functions. The abstraction combines algebraic decomposition with optimization-based bounding algorithms to provide arbitrarily tight bounds on nonlinear functions.
    \item We provide an efficient method for encoding polyhedral enclosures as mixed integer linear programs (MILPs).
    Our encoding addresses the limitations of existing approaches by exploiting problem structure to help minimize the encoding size.
    \item We define a novel algorithm for forward reachability and discuss several optimizations.
    \item We implement and evaluate our algorithm and show that it performs better than both OVERTVerify and CORA on a set of neural feedback system benchmarks.
\end{itemize}
The rest of the paper is organized as follows.  \Cref{sec:Back} covers background material, \cref{sec:Prob} defines the problem of neural feedback system verification, and \cref{sec:PolyEnc} introduces the theory of polyhedral enclosures.  Then, in \cref{sec:alg}, we define the OvertPoly algorithm and discuss its implementation.  \Cref{sec:eval} reports on our experimental evaluation, and \cref{sec:conclusion} summarizes our conclusions.

%% file: figures/pipeline.tex
\begin{tikzpicture}[
    >=Stealth,
    node distance=0.6cm and 1.1cm,
    % Main pipeline boxes
    stage/.style={
        rectangle, rounded corners=3pt,
        draw=#1!70!black, line width=0.7pt,
        fill=#1!8,
        minimum height=1.8cm, minimum width=2.6cm,
        align=center,
        font=\small\sffamily,
        inner sep=5pt,
    },
    % Stage labels (titles above boxes)
    stagelabel/.style={
        font=\footnotesize\sffamily\bfseries,
        text=#1!80!black,
        anchor=south,
    },
    % Internal detail text
    detail/.style={
        font=\scriptsize\sffamily,
        text=black!70,
        align=center,
    },
    % Arrows between stages
    pipe/.style={
        -{Stealth[length=5pt, width=4pt]},
        line width=0.8pt,
        color=black!50,
    },
    % Feedback arrow
    feedback/.style={
        -{Stealth[length=5pt, width=4pt]},
        line width=0.7pt,
        color=black!40,
        dashed,
    },
    % Small annotation
    annot/.style={
        font=\tiny\sffamily,
        text=black!55,
    },
]

% ---- Stage 1: Unicycle Car Model in Initial Region I ----
% Initial region box
\node[rectangle, rounded corners=3pt,
      draw=inputblue!70!black, line width=0.7pt,
      fill=inputblue!4,
      minimum height=2.5cm, minimum width=3.2cm,
      inner sep=5pt, dashed] (nfs) {};
\node[font=\footnotesize, text=inputblue!80!black, anchor=north east,
      inner sep=2pt] at (nfs.north east) {$X_t$};
\node[stagelabel=inputblue, above=3pt of nfs.north] {System $\mathcal{D}$};

% Draw unicycle inside the container
\begin{scope}[shift={($(nfs.center)+(-0.715, -0.715)$)}, scale=0.55]
    % Car body
    \filldraw[thick, rounded corners=0.2cm, fill=inputblue!18, draw=inputblue!60,
              rotate around={-45:(1.3,1.3)}]
              ({1.3 - 0.38}, {1.3 - 0.95}) rectangle ({1.3 + 0.38}, {1.3 + 0.95});
    \fill[inputblue!70] (1.3,1.3) circle (0.06cm);
    % v arrow: ahead of car nose (along heading ~45deg)
    \draw[black!65, thick, ->] ($(1.3,1.3)+(45:1.15)$) -- ++(45:0.65) node[right, font=\tiny] {$v$};
    % omega arc: upper-left
    \draw[black!50, thin, ->] ($(1.3,1.3)+(120:1.2)$) arc (120:180:0.45)
          node[pos=0.5, left, font=\tiny] {$\omega$};
    % x arrow: lower-right of car
    \draw[black!50, thin, ->] ($(1.3,1.3)+(-10:1.2)$) -- ++(0:0.6) node[above, font=\tiny] {$x$};
    % y arrow: upper-right of car
    \draw[black!50, thin, ->] ($(1.3,1.3)+(80:1.2)$) -- ++(90:0.6) node[right, font=\tiny] {$y$};
\end{scope}

% ---- Stage 2: Polyhedral Enclosures ----
\node[stage=boundgreen, right=1.3cm of nfs] (bound) {
    \textbf{Polyhedral}\\[1pt]
    \textbf{Enclosures}
};
\node[detail, below=0.5pt of bound.south, anchor=north] (bound-detail) {
   
};
% \node[stagelabel=boundgreen, above=3pt of bound.north] {Abstraction};

% ---- Stage 3: MILP Encoding ----
\node[stage=milporange, right=2.0cm of bound] (milp) {
    \textbf{MILP}\\[1pt]
    \textbf{Encoding}
};
\node[detail, below=0.5pt of milp.south, anchor=north] (milp-detail) {
    $M_0 \,{\cup}\, \{M_1, \ldots, M_n\}$
};

% ---- Stage 4: Reachable Set (lower level, below MILP) ----
\node[stage=reachred, below right=1.2cm and 0.6cm of milp] (reach) {
    \textbf{Reachable Set}\\[1pt]
    $\hat{X}_{t+1}$
};
\node[detail, below left=0.3cm and 0.5cm of reach.south, anchor=north] (reach-detail) {
    $\hat{X}_{t+1}$
};

\node[detail, below left=0.3cm and 6.5cm of reach.south, anchor=north] (reach-detail) {
    $\hat{X}_t \gets \hat{X}_{t + 1}$
};
% \node[stagelabel=reachred, left=8pt of reach.west, anchor=east] {Output};

% ---- Arrows between stages ----
\draw[pipe] (nfs.east) -- (bound.west);
\draw[pipe] (bound.east) -- (milp.west);
% MILP to Reachable Set (hinged: right then down)
\draw[pipe] (milp.east) -| (reach.north);
\node[annot, right=3pt] at ($(reach.north)+(0, 0.5)$) {solve};

% ---- Sub-annotations on arrows ----
% First arrow: F,X_t below, miniature dynamics surface above
\node[annot, below=1pt] at ($(nfs.east)!0.5!(bound.west)$) {$F, X_t$};
% Miniature nonlinear surface above the first arrow
\begin{scope}[shift={($(nfs.east)!0.5!(bound.west)+(0, 0.5)$)}, scale=0.35]
    % Back curve
    \draw[black!25, thin] (-1.0, 0.1) .. controls (-0.6, 0.7) and (-0.2, -0.2) .. (0, 0.15)
                          .. controls (0.2, 0.5) and (0.6, 0.6) .. (1.0, 0.25);
    % Front curve
    \draw[black!50, semithick] (-1.0, -0.3) .. controls (-0.6, 0.35) and (-0.2, -0.55) .. (0, -0.15)
                          .. controls (0.2, 0.2) and (0.6, 0.3) .. (1.0, -0.05);
    % Edge connectors
    \draw[black!20, very thin] (-1.0, 0.1) -- (-1.0, -0.3);
    \draw[black!20, very thin] (1.0, 0.25) -- (1.0, -0.05);
    % Fill
    \fill[inputblue!6, opacity=0.5]
        (-1.0, 0.1) .. controls (-0.6, 0.7) and (-0.2, -0.2) .. (0, 0.15)
                     .. controls (0.2, 0.5) and (0.6, 0.6) .. (1.0, 0.25)
        -- (1.0, -0.05) .. controls (0.6, 0.3) and (0.2, 0.2) .. (0, -0.15)
                         .. controls (-0.2, -0.55) and (-0.6, 0.35) .. (-1.0, -0.3) -- cycle;
\end{scope}
\node[annot, above=18pt] at ($(nfs.east)!0.5!(bound.west)$) {$f_i \in F$};
% Second arrow: B_i^t below, curve with polyhedral bounds above
\node[annot, below=1pt] at ($(bound.east)!0.5!(milp.west)$) {$B_t$};
% Miniature curve with piecewise-linear upper/lower bounds
\begin{scope}[shift={($(bound.east)!0.5!(milp.west)+(0, 0.55)$)}, scale=0.55]
    % --- Upper polyhedral bound (red, two rows + fill) ---
    \draw[reachred!40, thin] (-1.0, 0.3) -- (-0.3, 0.65) -- (0.2, 0.15) -- (0.6, 0.7) -- (1.0, 0.4);
    \draw[reachred!65, thin] (-1.0, 0.1) -- (-0.3, 0.45) -- (0.2, -0.05) -- (0.6, 0.5) -- (1.0, 0.2);
    \draw[reachred!20, very thin] (-1.0, 0.3) -- (-1.0, 0.1);
    \draw[reachred!20, very thin] (1.0, 0.4) -- (1.0, 0.2);
    \fill[reachred!5, opacity=0.4]
        (-1.0, 0.3) -- (-0.3, 0.65) -- (0.2, 0.15) -- (0.6, 0.7) -- (1.0, 0.4)
        -- (1.0, 0.2) -- (0.6, 0.5) -- (0.2, -0.05) -- (-0.3, 0.45) -- (-1.0, 0.1) -- cycle;
    % --- True nonlinear surface (black, two rows + fill) ---
    \draw[black!25, thin] (-1.0, 0.05) .. controls (-0.5, 0.5) and (0.0, -0.25) .. (0.5, 0.25)
                          .. controls (0.7, 0.4) and (0.9, 0.15) .. (1.0, 0.1);
    \draw[black!50, semithick] (-1.0, -0.15) .. controls (-0.5, 0.3) and (0.0, -0.45) .. (0.5, 0.05)
                          .. controls (0.7, 0.2) and (0.9, -0.05) .. (1.0, -0.1);
    \draw[black!15, very thin] (-1.0, 0.05) -- (-1.0, -0.15);
    \draw[black!15, very thin] (1.0, 0.1) -- (1.0, -0.1);
    \fill[black!3, opacity=0.3]
        (-1.0, 0.05) .. controls (-0.5, 0.5) and (0.0, -0.25) .. (0.5, 0.25)
                      .. controls (0.7, 0.4) and (0.9, 0.15) .. (1.0, 0.1)
        -- (1.0, -0.1) .. controls (0.9, -0.05) and (0.7, 0.2) .. (0.5, 0.05)
                        .. controls (0.0, -0.45) and (-0.5, 0.3) .. (-1.0, -0.15) -- cycle;
    % --- Lower polyhedral bound (green, two rows + fill) ---
    \draw[boundgreen!40, thin] (-1.0, -0.2) -- (-0.3, -0.0) -- (0.2, -0.4) -- (0.6, 0.0) -- (1.0, -0.15);
    \draw[boundgreen!70, thin] (-1.0, -0.4) -- (-0.3, -0.2) -- (0.2, -0.6) -- (0.6, -0.2) -- (1.0, -0.35);
    \draw[boundgreen!20, very thin] (-1.0, -0.2) -- (-1.0, -0.4);
    \draw[boundgreen!20, very thin] (1.0, -0.15) -- (1.0, -0.35);
    \fill[boundgreen!5, opacity=0.4]
        (-1.0, -0.2) -- (-0.3, -0.0) -- (0.2, -0.4) -- (0.6, 0.0) -- (1.0, -0.15)
        -- (1.0, -0.35) -- (0.6, -0.2) -- (0.2, -0.6) -- (-0.3, -0.2) -- (-1.0, -0.4) -- cycle;
\end{scope}

% ---- Neural controller annotation (feeds into MILP from above-right) ----
\node[rectangle, rounded corners=2pt, draw=inputblue!60, fill=inputblue!6,
      minimum height=0.7cm, minimum width=1.6cm, align=center,
      font=\scriptsize\sffamily, inner sep=3pt]
      (nn) at ($(milp.north)!0.65!(milp.north east) + (0, 1.2)$) {NN Controller $u$};
% Arrow: down from nn, landing on milp's north edge
\coordinate (nn-land) at ($(milp.north)!0.65!(milp.north east)$);
\draw[pipe, inputblue!50] (nn.south) -- (nn-land);
% \node[stagelabel=milporange, above=3pt of milp.north west, anchor=south] {Formulation};

% ---- Safety check annotation (right of reachable set) ----
\node[rectangle, rounded corners=2pt, draw=safeteal!70, fill=safeteal!8,
      minimum height=0.9cm, minimum width=2.0cm, align=center,
      font=\scriptsize\sffamily, inner sep=4pt,
      right=1.0cm of reach]
      (safe) {\textbf{Safety Check}};
\draw[pipe, safeteal!60] (reach.east) -- (safe.west);
%\node[annot, safeteal!70!black, below=3pt] at (safe.south) {$\hat{\tau} \cap A = \emptyset$?};

% ---- Feedback loop (lower level, back to Polyhedral Enclosures) ----
\coordinate (fb-left) at ($(bound.south)+(0, -2.0)$);
\coordinate (fb-right) at ($(reach.south)+(0, -0.8)$);
\draw[feedback]
    (reach.south) -- (fb-right)
    -| (fb-left)
    -- (bound.south);

% ---- Feedback annotations ----
% \node[annot] at ($(fb-left)!0.5!(fb-right)$) {%
%     $t \leftarrow t{+}1$ \;\textcolor{black!40}{(concrete or symbolic reachability)}%
% };

\end{tikzpicture}

%% file: sections/notation.tex
We denote the set of integers as \ints, the set of natural numbers (integers greater than zero) as \nats, the set of real numbers as \reals, and 
the set of non-negative real numbers as $\reals_+$.
If \State is a set, we denote the \emph{power set} of \State (i.e., the set of all subsets of \State) as $ \powset{\State} $.  
We use $ [i.. j] $ to represent the set $ \{z\in\ints \mid i\le z \le j\} $,  $[i,j]$ to represent the set $\{r\in\reals \mid i\le r\le j\}$, and $[n]$ to abbreviate $[1..n]$. If $r \in \reals \setminus \ints$, we write use $\lceil r \rceil$ to denote $\operatorname{ceil}(r)$, and $\lfloor r \rfloor$ to denote the corresponding $\operatorname{floor}$ operator.

If $ S $ is any finite \emph{sequence} $ (s_1,\dots,s_n) $, we write $ |S| $ to
denote $ n $, the length of the sequence, and $ S_i $ to denote the $i^{\mathit{th}}$ element of the sequence.
We write $S_{[i..j]}$ for the sequence $(s_i,\dots,s_j)$  and $ S \circ S' $ for the sequence obtained by
appending the sequence $ S' $ to the end of $ S $.
If a sequence is used where a set is expected, the meaning is the set of elements in the sequence (e.g., $s\in S$ means that $s$ occurs in the sequence $S$).
We use the same notation for both vectors and sequences and treat them as interchangeable. 

We use \vleq and \vgeq to describe element-wise inequalities for finite sequences, defined as follows.  If $x$ and $y$ are sequences of size $n$, then so is $x \vleq y$, with $(x \vleq y)_i = 1$, if $x_i \le y_i$, and $0$ otherwise. We define $x \vgeq y$ similarly. We write $x \cdot y$ for the dot product
(i.e. sum of element-wise products)
of vectors $x$ and $y$.  
We write $ \vec{0} $ for the zero vector, $ \vec{1} $ for the vector of all ones, and $ \unit{i} $ for the $ i^{\mathit{th}} $ unit vector. The sizes of these vectors will be left implicit when it is clear from context. A \emph{convex combination vector} $\vec{\theta}$ is a vector whose entries are non-negative and sum to 1, i.e., $\vec{\theta} \cdot \vec{1} = 1$ and $ \vec{\theta} \vgeq \vec{0} = \vec{1}$.

For a function $ f:\stateset \to \reals$ and a set $\stateset' \subseteq \stateset$, we define the \emph{image} of $\stateset'$ under $f$ as
  $f(\stateset') := \{f(x) \mid x \in \stateset'\}$.  Similarly, if $S$ is a sequence, then $f(S)$ is the sequence $(f(S_1),f(S_2),\dots)$.
For $\stateset' \subseteq \stateset$, we define the \emph{restriction} of $f$ to $\stateset'$ as the function $f^{\stateset'} : \stateset' \to \reals$ such that $f^{\stateset'}(x) = f(x)$ for every $x\in\stateset'$.

%% file: sections/polyhedra.tex
Let \pset be a set of $k+1$ points: $\pset = \{\vec{p}_0,\dots,\vec{p}_k\}$, with each $\vec{p}_i\in\reals^n$. 
 The \emph{convex hull} of $\pset$ is
    \begin{equation}
        \conv(\pset) 
        = \{\vec{\theta}_0 \vec{p}_0 + \ldots + \vec{\theta}_k \vec{p}_k \mid \vec{\theta} \cdot \vec{1} = 1, \vec{\theta} \vgeq \vec{0} \}.
    \end{equation}
The points in $\pset$ are \emph{affinely independent} iff the set $\{\vec{p}_1-\vec{p}_0,\dots,\vec{p}_k-\vec{p}_0\}$ is linearly independent.  The \emph{polyhedron formed by} a set of points $\pset$ is just $\conv(\pset)$.  A subset of $\reals^n$ is a polyhedron if it is the convex hull of a finite set of points in $\reals^n$.%
\footnote{Sometimes the term ``polyhedron'' is reserved for three-dimensional objects, with the generalization to arbitrary dimensions called a ``polytope.'' We use ``polyhedron'' also for the general case \citep{boyd2004convex,ziegler2012lectures}.}

\begin{definition}[$ k $-Simplex]
    A polyhedron $\simp$ is a $k$-simplex if it is the convex hull of $k+1$ affinely independent points.  A polyhedron is a simplex if it is a $k$-simplex for some $k$, and $k$ is called its \emph{dimension}.
\end{definition}

\noindent
If $\simp$ is a $k$-simplex, let $ \svert(\simp)$, the \emph{vertices} of \simp, denote the (unique) set of $k+1$ points $\pset$ such that $\simp=\conv(\pset)$. A \emph{face} of $\simp$ is the convex hull of any non-empty subset of $\svert(\simp)$.

\begin{definition}[Simplicial $k$-Complex]
A \emph{simplicial complex} $\scomp$ is a set of simplices such that:
\begin{itemize}
\item Every face of a simplex in \scomp is also in \scomp
\item Every non-empty intersection of two simplices $\simp_1,\simp_2\in\scomp$ is a face of both $\simp_1$ and $\simp_2$
\end{itemize}
\end{definition}

\noindent
\scomp is a \emph{pure simplicial $k$-complex} if the largest dimension of any simplex in \scomp is $k$ (also called the \emph{dimension} of \scomp) and if every simplex in \scomp of dimension less than $k$ is a face of some simplex in \scomp of dimension $k$.

\begin{definition}[Full-Dimensional]
Let \pset be a finite set of points in $\reals^n$.  We say that \pset is \emph{full-dimensional} if it contains $n+1$ affinely independent points.
\end{definition}

\begin{definition}[Point Set Triangulation]
If \pset is a finite, full-dimensional set of points in $\reals^n$, then a pure simplicial $n$-complex \scomp is a \emph{triangulation of $\pset$} if $\pset = \bigcup_{\simp\in\scomp} \svert(\simp)$ and $\conv(\pset) = \bigcup_{\simp\in\scomp} \simp$.
\end{definition}

Let $C$ be a closed $n$-ball.  We use $C^O$ to denote the corresponding open $n$-ball and $C^S$ to denote the hypersphere that forms the surface of $C$.  The vertices---with respect to a set \pset of points---of $C$ are defined as $V_\pset(C) = C \cap \pset$.  For a polyhedron $\poly$, the \emph{circumsphere} of $\poly$ (when it exists) is a hypersphere that touches all of the vertices of $\poly$.  Note that circumspheres always exist for simplices and for hyperrectangles.

Now, let $P$ be a finite, full-dimensional set of points in $\reals^n$, let $\scomp$ be a triangulation of $\pset$, and let $\simp$ be an $n$-simplex in \scomp.  We define $C(\simp)$ to be the $n$-ball whose boundary is the circumsphere of $\simp$.  $\simp$ satisfies the \emph{Delaunay condition} and is called a \emph{Delaunay simplex of $\pset$} if $V_\pset(C(\simp)^O) = \emptyset$ (i.e., the only points from $\pset$ contained in $C(\simp)$ are on its surface).  \scomp is a \emph{Delaunay triangulation} if every $n$-simplex in \scomp satisfies the Delaunay condition.  A set of points in $\reals^n$ has a Delaunay triangulation of dimension $n$ iff it is full-dimensional.

For a point $\vec{x}\in\conv(P)$, let $\simp_{\Delta}(\vec{x})$ be the $n$-simplex in $\Delta$ containing $\vec{x}$ (if $\vec{x}$ is in more than one $n$-simplex, which can only occur when it is on a face, we assume $\simp_{\Delta}(\vec{x})$ chooses one of the simplices in a deterministic way).  We define $\theta_{\Delta}(\vec{x})$ to be the convex combination vector such that $\vec{x} = \theta_{\Delta}(\vec{x}) \cdot \svert(\simp_{\Delta}(\vec{x}))$.
\begin{definition}[Grid]
A \emph{grid of dimension $n$} is the Cartesian product of $n$ finite subsets of $\reals$, each containing at least two elements. 
Given a grid $G$, we define $G_i$ as the projection of $G$ onto dimension $i$, so that $G = G_1 \times \dots \times G_n$. The \emph{domain} of the grid is defined as $\dom(G) = \conv(G).$

A \emph{grid cell} of $G$ is a subset of $G$ whose convex hull is an $n$-dimensional hyperrectangle $R$, such that no points of $G$ other than those forming the vertices of $R$ are contained in $R$.  For $\state\in G$, $\cells(G,\state)$ denotes the set of all grid cells containing $\state$, i.e., $\cells(G,\state) = \{\State \mid \State \text{ is a grid cell of } G \text{ and } \state \in \State\}.$
\end{definition}

%% file: sections/overt.tex
To compute tight upper and lower bounds for nonlinear one-dimensional functions, we use a method introduced in the OVERT algorithm~\citep{sidrane2022overt}. 
For a function $f(\state)$ that is convex on the interval $(a,b)$, an upper bound can be constructed by symbolically partitioning $(a,b)$ into $m>1$, subintervals with endpoints $(s_0 = a, s_1, \ldots, s_{m -1}, s_m = b)$ and defining secant lines from $(s_{i-1},f(s_{i-1}))$ to $(s_{i},f(s_{i}))$ for $i\in[m]$.  The point locations $s_i $ are then optimized to minimize the area between the secant lines and the function $ f(\state)$.
A lower bound can be constructed similarly by partitioning $(a,b)$ into subintervals with endpoints $s_i$, $i\in[0..m]$ and defining a sequence of line segments from $(s_{i},t_{i})$ to $(s_{i+1},t_{i+1})$, $i\in[0..m-1]$ in such a way that $t_0=f(a), t_m=f(b)$, each segment from $(s_i,t_i)$ $(s_{i+1},t_{i+1})$, $i\in[m-2]$ is tangent to $f(\state)$, and the area between the line segments and the function is once again minimized.
Bounds for a univariate function $ f(\state) $ that is concave over an interval are computed analogously, using a series of secants for lower bounds and a series of tangents for upper bounds. The tightness of the bounds can be adjusted by modifying the parameter $ m $.  Bounds over an arbitrary interval can be computed by first decomposing the function into intervals of uniform convexity and then composing the bounds obtained for each region.  Details of these methods are provided by \citet{sidrane2022overt}.

%% file: sections/mipVerify.tex
A \emph{feed-forward neural network} (FNN) is a function constructed by composing linear and nonlinear operations. These operations are organized in \emph{layers}, where each layer applies a linear transformation and (optionally) a nonlinear operation. The elements within a layer are referred to as \emph{neurons}, and the nonlinear operations are known as \emph{activation functions}. In this work, we focus on the Rectified Linear Unit (\ReLU) activation function, defined as $\ReLU(x) = \max(x,0)$. We refer to FNNs using only this activation function as \emph{\ReLU-activated feed-forward neural networks}.

\ReLU activated FNNs specify piecewise-linear functions and can therefore be formulated as mixed-integer linear programs (MILPs). Following the technique used by the MIPVerify algorithm~\citep{tjeng2017evaluating}, we encode \ReLU constraints as follows.  Let
$l,u \in \mathbb{R}$ be real numbers intended to bound $x$, with $l \leq 0 \leq u$,\footnote{Note that in the case that $l\leq u < 0$ or $0 < l \leq u$, a much simpler linear encoding is possible.}
and let
$z$ be a binary variable. Then both the \ReLU function
$y = \max(x,0)$ and the bounds $l \leq x \leq u$ can be represented by the constraints
\begin{gather}
    y \le x - l(1 - z), \quad
    y \ge x, \quad
    y \le u \cdot z, \quad
    x \in \reals, \quad y \in \reals_+, \quad z \in \{0,1\}.
    \label{eq:mipverify}
\end{gather}
To encode a \ReLU-activated FNN, the linear layers are expressed as standard linear constraints, while the nonlinear \ReLU constraints use~\Cref{eq:mipverify}.

%% file: figures/nfs.tex
\begin{tikzpicture}[>=latex, text height=1.5ex, text depth=0.25ex]
    % Nodes
    \node (plant) [draw, rectangle, minimum width=1.5cm, minimum height=1cm, label=above:{System Dynamics}] {$f(x)$};
    \node (controller) [draw, below=1.5cm of plant, rectangle, minimum width=1.5cm, minimum height=1cm, label=below:{NN Controller}] {$\pi(y)$};
    \node (input) [draw, left= 2cm of plant, circle]{$+$};
    \node (noise) [above=0.5cm of input]{$\epsilon$};
    \node (observer) [draw, below right=0.25cm and 1.5cm of plant, rectangle, minimum width=1.5cm, minimum height=1cm, label=left:{Sensor}] {$o(x)$};

    % Edges
    \draw[->] (noise) -- (input);
    \draw[->] (controller) -| (input)node[pos=0.3, yshift=7pt]{$u$};
    \draw[->] (input) -- (plant);
    \draw[->] (plant) -| (observer)node[pos=0.3, yshift=7pt]{$x$};; 
    \draw[->] (observer) |- (controller) node[midway, xshift=7pt, yshift=7pt]{$y$};
\end{tikzpicture}

%% file: sections/problem.tex
We define a discrete-time neural feedback system \dynsys as the tuple $ \langle n, \init, \trans, \noise,\! \ctrl, \timestep, \horizon, \reach, \avoid\rangle$, %At any time step $ t $ , 
where $ n\in\nats $ is the \emph{dimension} of the system (i.e., every state of the system is an element of $ \reals^n $ ), $ \init\subseteq\reals^n $ is the set of \emph{initial states}, $ \trans $ is a sequence $ (f_1,\dots,f_n) $ of \emph{state update functions} with $ f_i : \reals^n\to\reals $ , $ \noise\subseteq\reals^n $ is a perturbation error set (i.e., a set from which an error term may be introduced when computing the next state), $ \ctrl:\reals^n\to\reals^n $ is the \emph{control function}, $ \timestep\in\reals_+ $ is the \emph{time step size}, $ \horizon\in\nats $ is the \emph{number of time steps}, $ \reach\subseteq\reals^n $ is the set of \emph{goal states}, and $ \avoid:[0..T] \to 2^{\reals^n} $ is a function from time steps to the set of \emph{avoid states} (i.e., unsafe states) at that time step.

States evolve over a sequence of \horizon 
discrete time steps, each of duration \timestep, 
and the \emph{time horizon} is $\timestep \cdot \horizon$.  If $ \state\in\reals^n $ is a system state, the next-state function $ \nextF^{\dynsys} : \reals^n \to 2^{\reals^n} $ defines the set of possible next states (it is a set because of the nondeterminism introduced by the error term) as follows.  For each $ i\in[1..n] $ ,
\begin{align}
    \nextF^{\dynsys}(\state) = \{ \state + \trans(\state) + \ctrl(\state) + \vec{\epsilon}\big)\cdot \delta \mid \epsilon \in \noise\}
    \label{eq:diffEq}
\end{align}
We assume that each $ f_i \in \trans$ is from the class of Lipschitz continuous multivariate functions composed of rational operations (i.e., $+, -, \times, \div$) on univariate elementary functions.\footnote{Following the convention of~\citep{muller2006elementary}, we define elementary functions as the trigonometric functions, their inverses, the exponential functions, and logarithmic functions.}
We call these functions \emph{extended rational nonlinear functions}.\footnote{While we only consider a subset of nonlinear functions, the Kolmogorov-Arnold representation theorem~\citep{kuurkova1991kolmogorov} suggests that our approach can be generalized to arbitrary nonlinear functions.}
We also assume that the neural network controller $\ctrl$ is a multilayer perceptron with $ n $ inputs, $ n $ outputs, and ReLU activations.

A trajectory $ \traj^{\dynsys}(\State_0)$, where $ \State_0\subseteq\init$, is the sequence of sets of states $ (\State_0, \dots, \State_T)$, where $ \State_{i} = \nextF^{\dynsys}(\State_{i-1}) $ for $ i\in[1..T] $.
A system \dynsys is \emph{safe} if it satisfies the following reach-avoid properties:
\begin{align}
    &\forall\, \state_0\in\init.\: \exists \, t \in [0..T].\: \traj^{\dynsys}(\{\state_0\})_t \subseteq \reach , \label{eq:reach_prop}\\
    &\forall \, t \in [0..T].\: \traj^{\dynsys}(\init)_t \cap \avoid(t) = \emptyset. \label{eq:avoid_prop}
\end{align} 
Property \ref{eq:reach_prop} is a \textit{reach} property.  It states that every trajectory starting from some state in the initial state set reaches the goal set (specified by $\reach$) within the specified time horizon.  On the other hand, Property \ref{eq:avoid_prop} is an \textit{avoid} property, requiring that the system avoids any unsafe states ($\avoid$) at each time within the given time horizon.
%\cbtodo{Let's move all examples to the appendix.}
An illustrative example can be found 
below.
\input{sections/examples/unixample}

%% file: sections/examples/unixample.tex
\subsection{Illustrative Example: Unicycle Car Model}\label{ssec:uni}
As a running example, we use a discrete-time version of the unicycle car model example from the 2023 ARCH competition~\citep{lopez2023arch}. 
The car is modeled with four variables, representing the $ x $ and $ y $ coordinates in a plane, the steering angle ($\omega$), and the velocity magnitude $v$.
Formally, we define $ \uni = \langle n^{\uni}, \init^{\uni}, \trans^{\uni}, \noise^{\uni}, \ctrl^{\uni}, \timestep^{\uni}, \horizon^{\uni}, \reach^{\uni}, \avoid^{\uni}\rangle$, where $ n^{\uni}=4$, $ \init^{\uni} = [9.5,9.55] \times [-4.5,-4.45] \times [2.1,2.11] \times [1.5,1.51]$, $ \trans^{\uni}(\state) = (\state_4\cos(\state_3), \state_4\sin(\state_3),0,0)$, $ \noise = \{0\}\times\{0\}\times\{0\}\times [-10^{-4},10^{-4}]$, $ \ctrl^{\uni} $ 
is computed by a neural network with one hidden layer with 500 neurons and four outputs, the first two of which are set to the constant zero value (i.e., the controller only affects the velocity and steering), $ \timestep^{\uni} = 0.2 $ , $ \horizon^{\uni} = 50 $, $\reach^{\uni} =[-0.6,0.6] \times [-0.2,0.2] \times [-0.06,0.06] \times [-0.3,0.3]$, and $\avoid^{\uni}(t) = \emptyset$ for every $t \in [0..\horizon^{\uni}]$.  
 At each step, the system updates its $ x $ and $ y $ coordinates based on the steering and velocity control outputs, where the velocity ou   In our verification approach, we use \emph{polyhedral enclosures} to provide tight overapproximations of nonlinear functions.  In this section, we provide a formal definition of a polyhedral enclosure and then discuss how to construct and compose them.

%% file: sections/theory.tex
In our verification approach, we use \emph{polyhedral enclosures} to provide tight overapproximations of nonlinear functions.  In this section, we provide a formal definition of a polyhedral enclosure and then discuss how to construct and compose them.

\begin{definition}[Bounding Set] A \emph{bounding set} is a tuple $\vset = \langle n,\pset, L, U \rangle$, where $n \in \nats$, $\pset$ is a finite, full-dimensional set of points in $\reals^n$, and $L$ and $U$ are functions from \pset to $\reals$, such that $ L(\vec{p}) \leq U(\vec{p}) $ for all $ \vec{p} \in \pset. $ The \emph{domain} of \vset is defined as $\dom(\vset)=\conv(\pset)$.
\end{definition}
\input{sections/examples/bset_xample}

\begin{definition}[Polyhedron formed by Bounding Set]
Let $ \vset = \langle n,\pset,L,U \rangle $ be a bounding set. We define the \emph{vertices} of the bounding set as:
\[
V(\vset) := \{(\vec{p},L(\vec{p})) : \vec{p} \in \pset\} \cup \{(\vec{p},U(\vec{p})) : \vec{p} \in \pset\}.
\]
We define the ($n+1$-dimensional) \emph{polyhedron formed by $ \vset $ } as
\[
\poly(\vset) := \conv(V(\vset)).
\]
\end{definition}
\begin{figure}
    \centering
    \resizebox{0.9\textwidth}{!}{\input{figures/poly_enc_theory}}
    \caption{\footnotesize Visualizing the evolution of a polyhedral enclosure from a bounding set $\to$ a polyhedron $\to$ a polyhedral enclosure. }
    \label{fig:petheory}
\end{figure}
The polyhedron formed by $\vset^1$ from~\Cref{ex:boundingset} is a cube with
sides of length 10 centered at the origin.

\begin{definition}[Polyhedral Enclosure]
Let $ \vset = \langle n,\pset, L, U \rangle $ be a bounding set, let $\Delta$ be a Delaunay triangulation of $\pset$, and let $\Delta_n$ be the set of all $n$-simplices in $\Delta$. We define the bounding set associated with a simplex $\simp \in \Delta_n$ as:
\[
\vset_{\simp} := \langle n,\svert(\simp), L^{\svert(\simp)}, U^{\svert(\simp)} \rangle,
\]
\vspace{0.5em} 
\noindent
We define the \emph{polyhedral enclosure} formed by $\vset$ and $\Delta$ as:
\[
\polyenc(\vset,\Delta) := \bigcup_{S\in\Delta_n} \poly(\vset_S).
\]
\end{definition}
\input{sections/examples/polyenc_xample}

\begin{definition}[Function Enclosure]
Let $ \vset=\langle n,\pset,L,U\rangle$ be a bounding set.  A function $f : D \to \reals$, where $\dom(\vset) \subseteq D \subseteq \reals^n$, is \emph{enclosed} by $\vset$ if for every $\state\in\dom(\vset)$ and every Delaunay triangulation $\Delta$ of $\pset$, $(\state,f(\state)) \in 
\polyenc(\vset,\Delta)$.
\label{def:funcEnc}
\end{definition}
\input{sections/examples/func_xample}
\begin{figure}
    \centering
    \resizebox{0.6\textwidth}{!}{\input{figures/comp1}}
    \caption{\footnotesize Comparing the OVERTVerify and OvertPoly algorithms. OVERTVerify prioritizes precision by encoding rational operations as constraints in an optimization problem while OvertPoly retains structure using the notion of bound composition}
    \label{fig:comp}
\end{figure}
\subsection{Bounding Sets for Univariate Functions}
\modded{Clarified separation between our work and chelsea's. Also added a figure}
\label{sec:univar}
Let $f:\reals \to \reals$ be a function that is twice differentiable on the interval $[a,b]$.
We construct a bounding set enclosing $f$ as follows.

We partition the interval into $m>1$ subintervals of uniform convexity by identifying all points in $[a,b]$ where the second derivative of $f$ vanishes:
\[
z_0 := a, \, z_m := b, \, \text{and } \frac{d^2}{dx^2}f(z_i)=0 \text{ for } i=1,\ldots,m-1.
\]
Note that if $ \frac{d^2}{dx^2}f(\state) = 0 $ for \emph{all} $ x \in [a,b] $ , then we set $ m=1 $.
For each subinterval $ [z_i, z_{i+1}]$, we use the OVERT algorithm \citep{sidrane2022overt}
to compute piecewise-linear upper and lower bounds for $f$ on that interval. Relevant details about the OVERT algorithm are provided
 in \cref{sec:overt}.
Stitching all of them together, we get piecewise-linear functions $L$ and $U$ such that $\forall\,x.\:a\le x\le b \implies L(x) \le f(x) \le U(x)$. 

Now, we define the point set $\pset$ to be the set of all endpoints of line segments in the piecewise-linear functions $L$ and $U$ on the interval $[a,b]$ (including $a$ and $b$ themselves).
A bounding set enclosing $f$ is then:
$\vset := \langle 1,\pset, L^{\pset}, U^{\pset} \rangle$. 
\textbf{We differ from OVERTVerify in how we generalize to multivariate functions.} OVERTVerify bounds multivariate functions using implicit relationships between variables, treating rational operations as constraints and thereby discarding polyhedral structure in higher dimensions. In contrast, our approach prioritizes preserving this polyhedral structure. We achieve this by exploiting the notion of composition, which we describe in the following section.
\subsection{Composing Bounding Sets}
\label{sec:compose}
In this section, we introduce a notion of \emph{composition} for bounding sets.  Our goal is to show that composition preserves enclosure.
We start by defining \emph{lifting} and \emph{interpolation}, two operations on bounding sets that are needed to define composition.
\emph{Lifting} extends a function to a higher dimensional domain in such a way that the restriction to the original domain is the original function.  Lifting is parameterized by a set $S$, which identifies the indices of the arguments in the higher-dimensional function that correspond to arguments in the original unlifted function (the rest of the higher-dimensional arguments are ignored).

\begin{definition}[Lifted Function] 
Let $f : \reals^k \to \reals$ be a function, let $n > k$, and let $S=\{j_1,\dots,j_k\}$ be a subset of $[n]$ with $j_i < j_{i+1}$ for each $i\in[k-1]$.  We define \emph{$f$ lifted to $n$ by $S$}, $\lift{f}{S,n}{}$ as follows.  For each $\vec{y}\in \reals^n$, $(\lift{f}{S,n}{})(\vec{y}) = f(\vec{x})$, where $\vec{x}\in\reals^k$ and for each $i\in[k]$, $x_i = y_{j_i}$.
\end{definition}

\begin{example}(Lifted Function)
  Let $f$ be the function from~\Cref{ex:funcenc}, and let $S = \{3,4\}$.
  Then, if $\vec{y}=(y_1,y_2,y_3,y_4)$, $\lift{f}{S,4}{}(y) = y_4\cos(y_3)$.
\end{example}

\begin{definition}[Lifted Bounding Set]
Let $\vset := \langle k, \pset, L, U \rangle$ be a bounding set, let $n > k$, let $S=\{j_1,\dots,j_k\}$ be a subset of $[n]$ with $j_i < j_{i+1}$ for each $i\in[k-1]$, and let $\pl,\pu\in\reals^n$, with $\pl_i < \pu_i$ for $i\in[n]\setminus S$.  We define \emph{$\vset$ lifted to $n$ by $S$ from $\pl$ to $\pu$}, $\lift{\vset}{S,n}{\pl,\pu}$ as follows.  $\lift{\vset}{S,n}{\pl,\pu} = \langle n, \pset', \lift{L}{S,n}{}, \lift{U}{S,n}{}\rangle$, where $\pset' = \{\vec{p}{\:}' \in \reals^n \mid \exists\,\vec{p}\in\pset.\:\vec{p}{\:}'_i = \vec{p}_i \text{ if } i \in S \text{ and } \vec{p}{\:}'_i \in \{\pl_i,\pu_i\} \text{ otherwise }$.
\end{definition}

\noindent
Note that the purpose of $\pl$ and $\pu$ is to provide lower and upper bounds for the new dimensions added when lifting the points in $\pset$, to ensure that $P'$ is full-dimensional.

\begin{example}(Lifted Bounding Sets)
Let $\vset^1$ be the bounding set from~\cref{ex:boundingset}, let $S =
\{1,2\}$, let $\pu =
(0,0,5)$, and let $\pl = (0,0,-5)$.  Then,
$\lift{\vset}{S,3}{\pl,\pu}$ is a bounding set whose domain is the
cube with sides of length 10 centered at the origin, and the polyhedron formed by the lifted bounding set
is a hypercube in $\reals^4$ also centered at the origin.
\end{example}
The theorem below states that lifting preserves function enclosure. 
\begin{theorem} 
    If a bounding set $\vset=\langle k,\pset,L,U\rangle$ encloses a function $f$, then every lifted bounding set encloses the corresponding lifted function.  Formally, for every $n>k$, $S \subset [n]$ with $|S|=k$, and $\pl,\pu\in\reals^n$, with $\pl_i < \pu_i$ for $i\in[n]\setminus S$, \lift{\vset}{S,n}{\pl,\pu} encloses \lift{f}{S,n}{}.
    \label{thm:lift}
\end{theorem}
\begin{figure}
    \centering
    \resizebox{0.9\textwidth}{!}{\input{figures/lif_interp}}
    \caption{\footnotesize Visualizing the lifting and interpolation operations}
    \label{fig:liftint}
\end{figure}
\input{sections/proofs/thm1_proofs}

Lifting is one way to extend a bounding set.  Another is to add vertices to the bounding set.  This does not work in general, but it does work for bounding sets whose point sets are grids.  We thus focus on such bounding sets from here on.

\begin{definition}[Grid Expansion]
Let $G$ be an $n$-dimensional grid, and let $\vec{q}\in\dom(G)$.  Let $G'_i$ be the set $G_i \cup \{q_i\}$.  Then $G$ expanded by $\vec{q}$, written $G + \vec{q}$, is the grid $G'_1 \times \dots \times G'_n$.
\end{definition}

\begin{definition}[Interpolated Function]
Let $\pset\subset\reals^n$ be a grid, $f : \pset \to \reals$ a function, and $\Delta$ a Delaunay triangulation of $\pset$.  Let $\vec{q}\in\dom(\pset)$, and let $P' = P + \vec{q}$.
Then, \emph{$f$ interpolated by $\vec{q}$ using $\Delta$}, written $f \interp \vec{q}$ is the function $f' : \pset'  \to \reals$ defined as:

\[(f \interp \vec{q})(\vec{x}) = \begin{cases}
    f(\vec{x}) \text{ if } \vec{x} \in P, \text{ and}\\
    \theta_{\Delta}(\vec{x}) \cdot f(\svert(\simp_{\Delta}(\vec{x}))) \text{ otherwise.}
    \end{cases}
\]
\end{definition}

\begin{example}(Interpolated Function)
    Let $\vset^1$ be the bounding set from previous examples, and let
  $\Delta$ be the triangulation of $\pset_1$ whose 2-simplices are the
  triangles $T_1$, with vertices $\vec{p}_1, \vec{p}_2, \vec{p}_3$, and $T_2$, with vertices
  $\vec{p}_2, \vec{p}_3, \vec{p}_4$.
  Let $\vec{q}=(0,0)$, and let $P'_1 = P_1 + \vec{q} = P_1 \cup \{(-5,0),(0,0),(5,0),(0,-5),(0,5)\}$.
Then, let $U'_1 = U_1 \interp \vec{q}$.  Let $\vec{x} = (-5,0)$.  To compute $U'_1(\vec{x})$, we first note that $\simp_{\Delta}(\vec{x}) = T_1$, $\svert(T_1) = \{\vec{p}_1,\vec{p}_2,\vec{p}_3\}$, and $\theta(\vec{x}) = (0.5,0.5,0)$.  Thus, $U'_1(\vec{x}) = \theta(\vec{x}) \cdot U_1(\svert(T_1)) = (0.5,0.5,0) \cdot (5,5,5) = 5$.
%$(\vec{x}) = \theta \cdot U_1(\svert(T_1))$ is an interpolated function.
    \label{ex:interpfunc}
\end{example}

\begin{definition}[Interpolated Bounding Set]
Let $\vset = \langle n, \pset, L, U\rangle$ be a bounding set, with $\pset$ a grid, let $\Delta$ be a Delaunay triangulation of $\pset$, and let $\vec{q} \in \dom(\pset)$.  Then \emph{$\vset$ interpolated by $\vec{q}$ using $\Delta$}, written $\vset \interp \vec{q}$, is defined as:
\[\vset \interp \vec{q} = \langle n, \pset + \vec{q}, L\interp \vec{q}, U\interp\vec{q} \rangle.\]
\end{definition}
\begin{example}(Interpolated Bounding Set)
    Let $\vset^1$ be the bounding set from \cref{ex:boundingset}, $\Delta$ the triangulation from \cref{ex:interpfunc}, and $\vec{q} = (0,0)$. Then $\vset^1 \interp \vec{q} = \langle 2, \pset'_1, L_1\interp \vec{q}, U'_1 \rangle$, where $\pset'_1$ and $U'_1$ are as in \cref{ex:interpfunc}, and $L_1\interp \vec{q}$ is computed similarly..
\end{example}
\noindent
Before we can show that interpolation preserves enclosure, we need the following property of Delaunay triangulations for grids. 
\begin{lemma}
\label{lem:gridsimp}
Let $G$ be an $n$-dimensional grid, and let $\scomp$ be a Delaunay triangulation of $G$.  Then, for every $n$-simplex $\simp \in \scomp$, $C(\simp)$ is also the circumsphere of a grid cell of $G$.
\end{lemma}
\input{sections/proofs/lem1_proof}
A direct consequence of \Cref{lem:gridsimp} is that $G$ does not have a unique Delaunay triangulation. Since each $n-$simplex $\simp \in \scomp$ shares its circumsphere with a grid cell $\cell \in G$, then the other vertices of \cell can be combined with a subset of the vertices of \simp to form a Delaunay $n-$simplex, yielding a different triangulation. 

\begin{theorem}
Let $\vset = \langle n, \pset, L, U\rangle$ be a bounding set, with $\pset$ a grid.  If $\vset$ encloses a function $f$, then for every Delaunay triangulation $\Delta$ of $\pset$ and every point $\vec{q}\in\dom(\vset)$, $\vset \interp \vec{q}$ also encloses $f$.
\label{thm:intp}
\end{theorem}
\input{sections/proofs/thm2_proofs}

We are now ready to consider composing bounding sets. 
Let $\vset^{f}=\langle n^f,\pset^f,L^f,U^f\rangle$ and $\vset^{g}=\langle n^g, \pset^g, L^g, U^g\rangle$ be bounding sets, and suppose that $\vset^{f}$ encloses $f$ and $\vset^{g}$ encloses $g$.  We wish to show that certain compositions of $\vset^{f}$ and $ \vset^{g}$ enclose corresponding compositions of $f$ and $g$.  Before composing, however, we use lifting and interpolation to obtain bounding sets with \emph{identical point sets} for (lifted versions of) $f$ and $g$.

The first step is to lift $f$ and $g$ to some dimension $n$.  We have to decide how we want to embed $\reals^{n^f}$ and $\reals^{n^g}$ into $\reals^n$.  We can choose any $n$ such that $\max(n^f,n^g)\leq n \leq n^f+n^g$.  The embedding is done via sets $S^f$ and $S^g$.  $S^f\subseteq [n]$ lists the positions taken by the arguments to $f$ in a full list of $n$ variables, and the same is true for $S^g$.
Then, $f' = \lift{f}{S^f,n}{}$ and $g' = \lift{g}{S^g,n}{}$ are lifted versions of $f$ and $g$, respectively, that both map from $\reals^n$ to $\reals$.  To lift the bounding sets, we define $\pl$, $\pu$ to be vectors of size $n$ such that $\pl_i = \min(\pset^f_i \cup \pset^g_i)$ and $\pu_i = \max(\pset^f_i \cup \pset^g_i)$, for $i\in[n]$.  Now, let $\vset^{f'} = \langle n,\pset^{f'},L^{f'},U^{f'}\rangle = \lift{\vset^{f}}{S^f,n}{\pl,\pu}$ and $\vset^{g'} = \langle n,\pset^{g'},L^{g'},U^{g'}\rangle = \lift{\vset^{g}}{S^g,n}{\pl,\pu}$.  We know that $\vset^{f'}$ and $\vset^{g'}$ enclose $f'$ and $g'$, respectively, by \Cref{thm:lift}.

At this point, we need $\dom(\vset^{f'}) = \dom(\vset^{g'})$.  Notice that because of the way we chose $\pl$ and $\pu$, this will be the case unless there is some $j_{i^f}\in S^f$ and $j_{i^g}\in S^g$ such that $j_{i^f}=j_{i^g}$ and $\min(\pset^f_{i^f}) \not= \min(\pset^g_{i^g})$ or $\max(\pset^f_{i^f}) \not= \max(\pset^g_{i^g})$.   To ensure this doesn't happen, we assume that each dimension has fixed lower and upper bounds.  In particular, for neural feedback systems, we assume that the initial state set $\init$ bounds each individual state variable.

Next, we need both bounding sets to use the same point set. We achieve this by defining $\vset^{f''}$ as follows.  Let $\{\vec{p_1},\dots,\vec{p_m}\}$ be an enumeration of the points in $\pset^{g'} \setminus \pset^{f'}$.  Define $\vset^{f''}_0 = \vset^{f'}$, and define $\vset^{f''}_i = \vset^{f''}_{i-1} \interpi \vec{p_i}$, for $i\in [m]$, where $\Delta_i$ is an arbitrary Delaunay triangulation of the point set of $\vset^{f''}_{i-1}$.  We then set $\vset^{f''} = \vset^{f''}_m$.  We define $\vset^{g''}$ similarly as the result of inserting into the point set of $\vset^{g'}$ any missing points from the point set of $\vset^{f''}$.  The resulting bounding sets, $\vset^{f''}$ and $\vset^{g''}$ are defined over the same point set and still enclose $f'$ and $g'$, respectively, by~\Cref{thm:intp}.
We can now define bounding set composition.
\begin{figure*}[t!]
    \center
    \input{figures/bound}
    \caption{\footnotesize Overview of bounding algorithm for example problem. We first bound the univariate functions $ \cos(\mathbf{x}_3) $ and $ \mathbf{x}_4 $ , yielding piecewise-linear lower and upper bounds.  We then compose these enclosures using multiplication to obtain the polyhedral enclosure shown on the right.
    }
    \label{fig:bounds}
\end{figure*}
\begin{definition}[Linear Composition of Bounding Sets]
\label{def:lin_comp}
Let $\vset^f = \langle n, \pset, L^f, U^f \rangle$ and $\vset^g = \langle n, \pset, L^g, U^g \rangle$ be bounding sets, and let $\bowtie\ \in\{+,-\}$. We define $\vset^f \bowtie \vset^g := \langle n, \pset, L^{\fg}, U^{\fg} \rangle$, where, for all $\vec{x}\in\pset$,
    \begin{align*}
    &L^{\fg}(\vec{x}), U^{\fg}(\vec{x}) = \\
    &\begin{cases}
        L^f(\vec{x}) + L^g(\vec{x}), U^f(\vec{x}) + U^g(\vec{x}) & \text{if } \bowtie \, = \,+, \\[0.1ex]
        L^f(\vec{x}) - U^g(\vec{x}), U^f(\vec{x}) - L^g(\vec{x}) & \text{if } \bowtie \, = \,-, \\[0.1ex]
    \end{cases}
    \end{align*}
\end{definition}
\begin{theorem}
Let $\vset^f = \langle n, \pset, L^f, U^f \rangle$ and $\vset^g = \langle n, \pset, L^g, U^g \rangle$, and suppose that $\vset^f$ encloses $f$ and $\vset^g$ encloses $g$.  Then, for $\bowtie\ \in\{+,-\}$, $\vset^f \bowtie \vset^g$ encloses $f \bowtie g$.
\end{theorem}
\input{sections/proofs/thm3_proofs}

The situation is more complex for nonlinear operations.
We must define a notion of composition that preserves the enclosure property while retaining the (piecewise) linear structure of bounding sets.   Intuitively, we relax the bounds on each grid cell to obtain constant bounds, then perform the nonlinear operation  using interval arithmetic, and then take the weakest bound at each grid cell boundary.

\begin{definition}[Nonlinear Composition of Bounding Sets]
    Let $\vset^f = \langle n, \pset, L^f, U^f \rangle$ and $\vset^g = \langle n, \pset, L^g, U^g \rangle$ be bounding sets, and let $\bowtie\ \in\{\times,\div\}$.  Furthermore, if $\bowtie \,=\div$, assume that $0 \notin [L^g(\state),U^g(\state)]$ for every $\state\in \pset$. We define the nonlinear composition $\vset^f \bowtie \vset^g$ as $\langle n, \pset, L^{\fg}, U^{\fg} \rangle$, where $L^{\fg}$ and $U^{\fg}$ are defined as follows. First, for each grid cell $\State$ of $\pset$, define:
    \begin{align*}
        L^f_\State, U^f_\State & =\min(L^f(\State)), \max(U^f(\State))\\
L^g_\State,U^g_\State & =\min(L^g(\State)), 
        \max(U^g(\State))
   \end{align*}
   We next use interval arithmetic to define lower and upper bounds for the composition on each grid cell:
   \begin{align*}
        \mathcal{H}_{\bowtie}(\State) & = \{h_1 \bowtie h_2 \,|\, h_1 \in \{L^f_\State, U^f_\State\}, h_2 \in
        \{L^g_\State, U^g_\State\} \}\\
        L_{\bowtie}(\State) & = \min(\mathcal{H}_{\bowtie}(\State)),\ U_{\bowtie} =(\State)\max(\mathcal{H}_{\bowtie}(\State)).
    \end{align*}
Finally, for $\state\in\pset$, define:
\begin{align*}
    L^{\fg}(\state) & = \min(\{L_{\bowtie}(\State)\, \mid\,\State\in\cells(\pset,\state)\}),\\
    U^{\fg}(\state) & = \max(\{U_{\bowtie}(\State)\, \mid\,  \State\in\cells(\pset,\state)\}).
\end{align*}
\label{def:nln_comp}
\end{definition}
\begin{theorem}
$\!$Let $\vset^f\!\! =\! \langle n, \pset, L^f\!\!, U^f \rangle$ and $\vset^g\! =\! \langle n, \pset, L^g\!\!, U^g \rangle$, and suppose that $\vset^f$ encloses $f$ and $\vset^g$ encloses $g$.  Then, $\vset^f \times \vset^g$ encloses $f \times g$.  Furthermore, if for every Delaunay triangulation $\Delta$ of $\pset$, $\polyenc(\vset^g,\Delta) \cap (z(\state)=0) = \emptyset$ (i.e., the enclosure for $g$ does not intersect the plane corresponding to the constant zero function), then $\vset^f \div \vset^g$ encloses $f \div g$.
\end{theorem}
\input{sections/proofs/thm4_proofs}
\Cref{fig:bounds} illustrates the effect of applying these operations to our running example.

%% file: sections/examples/bset_xample.tex
\begin{example}(Bounding Set)
  Let $\vec{p}_1=(-5,-5)$, $\vec{p}_2=(-5,5)$, $\vec{p}_3=(5,-5)$, $\vec{p}_4=(5,5)$ be points in $\reals^2$
  and let $\pset_1=\{\vec{p}_1, \vec{p}_2, \vec{p}_3, \vec{p}_4\}$.  Then, let $\vset^1$ be the bounding set $\langle 2, \pset_1,L_1,U_1 \rangle$, with $U_1(\vec{p}) = 5$ and $L_1(\vec{p}) = -5$
  for every $\vec{p} \in \pset_1$. The domain of $\vset^1$,  $\dom(\vset^1)$, is a
  square with sides of length $10$ centered at the origin.
  \label{ex:boundingset}
\end{example}

%% file: figures/poly_enc_theory.tex
\begin{tikzpicture}[
    every node/.style={font=\small},
    >=stealth,
]

% === Data ===
% Points: p1=0, p2=1, p3=2, p4=3
% U_vals: 0.50, 1.35, 1.40, 0.65
% L_vals: 0.10, 0.90, 0.95, 0.20
% f(x) = sin(x) + 0.3

% === Scaling ===
% x: 1 unit = 1.1cm within each panel
% y: 1 unit = 2.2cm
% Panel spacing: 4.2cm between origins

\def\xscale{1.1}
\def\yscale{2.2}
\def\panelgap{5.2}

% Coordinates as macros
\def\pone{0}    \def\ptwo{1}   \def\pthree{2}  \def\pfour{3}
\def\Uone{0.50} \def\Utwo{1.35} \def\Uthree{1.40} \def\Ufour{0.65}
\def\Lone{0.10} \def\Ltwo{0.90} \def\Lthree{0.95} \def\Lfour{0.20}

% =====================================================
% Panel (a): Bounding set B
% =====================================================
\begin{scope}[shift={(0,0)}]
    % Title
    \node[anchor=south, font=\small] at (1.5*\xscale, 1.65*\yscale) {Bounding set $B$};
    \node[anchor=south east, font=\small\bfseries] at (-0.0*\xscale, 1.65*\yscale) {(a)};

    % x-axis
    \draw[thin] (-0.3*\xscale, 0) -- (3.4*\xscale, 0);
    % Tick marks and labels
    \foreach \p/\lab in {0/$p_1$, 1/$p_2$, 2/$p_3$, 3/$p_4$} {
        \draw[thin] (\p*\xscale, -0.06) -- (\p*\xscale, 0.06);
        \node[below, font=\footnotesize] at (\p*\xscale, -0.06) {\lab};
    }

    % U(p_i) - filled circles
    \fill (\pone*\xscale, \Uone*\yscale) circle (2.2pt);
    \fill (\ptwo*\xscale, \Utwo*\yscale) circle (2.2pt);
    \fill (\pthree*\xscale, \Uthree*\yscale) circle (2.2pt);
    \fill (\pfour*\xscale, \Ufour*\yscale) circle (2.2pt);

    % L(p_i) - open circles
    \draw[fill=white, thin] (\pone*\xscale, \Lone*\yscale) circle (2.2pt);
    \draw[fill=white, thin] (\ptwo*\xscale, \Ltwo*\yscale) circle (2.2pt);
    \draw[fill=white, thin] (\pthree*\xscale, \Lthree*\yscale) circle (2.2pt);
    \draw[fill=white, thin] (\pfour*\xscale, \Lfour*\yscale) circle (2.2pt);

    % Legend (boxed)
    \fill[black!8, draw=black!50, thin, rounded corners=2pt]
        (2.2*\xscale, 1.25*\yscale) rectangle (3.45*\xscale, 1.60*\yscale);
    \fill (2.4*\xscale, 1.48*\yscale) circle (2pt);
    \node[right, font=\scriptsize] at (2.55*\xscale, 1.48*\yscale) {$U(p_i)$};
    \draw[fill=white, thin] (2.4*\xscale, 1.35*\yscale) circle (2pt);
    \node[right, font=\scriptsize] at (2.55*\xscale, 1.35*\yscale) {$L(p_i)$};
\end{scope}

% =====================================================
% Panel (b): Polyhedron P(B)
% =====================================================
\begin{scope}[shift={(\panelgap, 0)}]
    % Title
    \node[anchor=south, font=\small] at (1.5*\xscale, 1.65*\yscale) {Polyhedron $\mathcal{P}(B)$};
    \node[anchor=south east, font=\small\bfseries] at (-0.00*\xscale, 1.65*\yscale) {(b)};

    % x-axis
    \draw[thin] (-0.3*\xscale, 0) -- (3.4*\xscale, 0);
    \foreach \p/\lab in {0/$p_1$, 1/$p_2$, 2/$p_3$, 3/$p_4$} {
        \draw[thin] (\p*\xscale, -0.06) -- (\p*\xscale, 0.06);
        \node[below, font=\footnotesize] at (\p*\xscale, -0.06) {\lab};
    }

    % Convex hull of all 8 vertices (sorted by angle)
    % Hull vertices (counterclockwise): 
    % (0,0.10), (3,0.20), (3,0.65), (2,1.40), (1,1.35), (0,0.50)
    \fill[black!12, draw=black, thin]
        (\pone*\xscale, \Lone*\yscale) --
        (\pfour*\xscale, \Lfour*\yscale) --
        (\pfour*\xscale, \Ufour*\yscale) --
        (\pthree*\xscale, \Uthree*\yscale) --
        (\ptwo*\xscale, \Utwo*\yscale) --
        (\pone*\xscale, \Uone*\yscale) -- cycle;

    % U(p_i) - filled circles
    \fill (\pone*\xscale, \Uone*\yscale) circle (2.2pt);
    \fill (\ptwo*\xscale, \Utwo*\yscale) circle (2.2pt);
    \fill (\pthree*\xscale, \Uthree*\yscale) circle (2.2pt);
    \fill (\pfour*\xscale, \Ufour*\yscale) circle (2.2pt);

    % L(p_i) - open circles
    \draw[fill=white, thin] (\pone*\xscale, \Lone*\yscale) circle (2.2pt);
    \draw[fill=white, thin] (\ptwo*\xscale, \Ltwo*\yscale) circle (2.2pt);
    \draw[fill=white, thin] (\pthree*\xscale, \Lthree*\yscale) circle (2.2pt);
    \draw[fill=white, thin] (\pfour*\xscale, \Lfour*\yscale) circle (2.2pt);

    % Label
    \node[font=\scriptsize, text=black!60] at (1.5*\xscale, 1.55*\yscale)
        {$\mathcal{P}(B) = \conv(V(B))$};
\end{scope}

% =====================================================
% Panel (c): Polyhedral enclosure E(B, Delta)
% =====================================================
\begin{scope}[shift={(2*\panelgap, 0)}]
    % Title
    \node[anchor=south, font=\small] at (1.5*\xscale, 1.65*\yscale) {Polyhedral enclosure};
    \node[anchor=south east, font=\small\bfseries] at (-0.00*\xscale, 1.65*\yscale) {(c)};

    % x-axis
    \draw[thin] (-0.3*\xscale, 0) -- (3.4*\xscale, 0);
    \foreach \p/\lab in {0/$p_1$, 1/$p_2$, 2/$p_3$, 3/$p_4$} {
        \draw[thin] (\p*\xscale, -0.06) -- (\p*\xscale, 0.06);
        \node[below, font=\footnotesize] at (\p*\xscale, -0.06) {\lab};
    }

    % Simplex S1: p1--p2
    \fill[black!12, draw=black, thin]
        (\pone*\xscale, \Lone*\yscale) --
        (\ptwo*\xscale, \Ltwo*\yscale) --
        (\ptwo*\xscale, \Utwo*\yscale) --
        (\pone*\xscale, \Uone*\yscale) -- cycle;

    % Simplex S2: p2--p3 (slightly darker)
    \fill[black!18, draw=black, thin]
        (\ptwo*\xscale, \Ltwo*\yscale) --
        (\pthree*\xscale, \Lthree*\yscale) --
        (\pthree*\xscale, \Uthree*\yscale) --
        (\ptwo*\xscale, \Utwo*\yscale) -- cycle;

    % Simplex S3: p3--p4
    \fill[black!12, draw=black, thin]
        (\pthree*\xscale, \Lthree*\yscale) --
        (\pfour*\xscale, \Lfour*\yscale) --
        (\pfour*\xscale, \Ufour*\yscale) --
        (\pthree*\xscale, \Uthree*\yscale) -- cycle;

    % Simplex labels (centered in each box)
    \node[font=\scriptsize, text=black!50] at (0.5*\xscale, 0.7125*\yscale) {$S_1$};
    \node[font=\scriptsize, text=black!50] at (1.5*\xscale, 1.15*\yscale) {$S_2$};
    \node[font=\scriptsize, text=black!50] at (2.5*\xscale, 0.80*\yscale) {$S_3$};

    % U(p_i) - filled circles
    \fill (\pone*\xscale, \Uone*\yscale) circle (2.2pt);
    \fill (\ptwo*\xscale, \Utwo*\yscale) circle (2.2pt);
    \fill (\pthree*\xscale, \Uthree*\yscale) circle (2.2pt);
    \fill (\pfour*\xscale, \Ufour*\yscale) circle (2.2pt);

    % L(p_i) - open circles
    \draw[fill=white, thin] (\pone*\xscale, \Lone*\yscale) circle (2.2pt);
    \draw[fill=white, thin] (\ptwo*\xscale, \Ltwo*\yscale) circle (2.2pt);
    \draw[fill=white, thin] (\pthree*\xscale, \Lthree*\yscale) circle (2.2pt);
    \draw[fill=white, thin] (\pfour*\xscale, \Lfour*\yscale) circle (2.2pt);

    % Label
    \node[font=\scriptsize, text=black!60] at (1.5*\xscale, 1.55*\yscale)
        {$\mathcal{E}(B, \Delta) = \bigcup_{S \in \Delta_n} \mathcal{P}(B_S)$};

\end{scope}

% =====================================================
% Panel (d): Function enclosure
% =====================================================
\begin{scope}[shift={(3*\panelgap, 0)}]
    % Title
    \node[anchor=south, font=\small] at (1.5*\xscale, 1.65*\yscale) {Function enclosure};
    \node[anchor=south east, font=\small\bfseries] at (-0.00*\xscale, 1.65*\yscale) {(d)};

    % x-axis
    \draw[thin] (-0.3*\xscale, 0) -- (3.4*\xscale, 0);
    \foreach \p/\lab in {0/$p_1$, 1/$p_2$, 2/$p_3$, 3/$p_4$} {
        \draw[thin] (\p*\xscale, -0.06) -- (\p*\xscale, 0.06);
        \node[below, font=\footnotesize] at (\p*\xscale, -0.06) {\lab};
    }

    % Per-simplex enclosure (light fill, dashed edges)
    \fill[black!12]
        (\pone*\xscale, \Lone*\yscale) --
        (\ptwo*\xscale, \Ltwo*\yscale) --
        (\ptwo*\xscale, \Utwo*\yscale) --
        (\pone*\xscale, \Uone*\yscale) -- cycle;
    \fill[black!12]
        (\ptwo*\xscale, \Ltwo*\yscale) --
        (\pthree*\xscale, \Lthree*\yscale) --
        (\pthree*\xscale, \Uthree*\yscale) --
        (\ptwo*\xscale, \Utwo*\yscale) -- cycle;
    \fill[black!12]
        (\pthree*\xscale, \Lthree*\yscale) --
        (\pfour*\xscale, \Lfour*\yscale) --
        (\pfour*\xscale, \Ufour*\yscale) --
        (\pthree*\xscale, \Uthree*\yscale) -- cycle;

    % Dashed enclosure outline
    \draw[black!45, thin, dashed]
        (\pone*\xscale, \Lone*\yscale) --
        (\ptwo*\xscale, \Ltwo*\yscale) --
        (\pthree*\xscale, \Lthree*\yscale) --
        (\pfour*\xscale, \Lfour*\yscale);
    \draw[black!45, thin, dashed]
        (\pone*\xscale, \Uone*\yscale) --
        (\ptwo*\xscale, \Utwo*\yscale) --
        (\pthree*\xscale, \Uthree*\yscale) --
        (\pfour*\xscale, \Ufour*\yscale);
    % Vertical dashed lines at simplex boundaries
    \draw[black!45, thin, dashed] (\ptwo*\xscale, \Ltwo*\yscale) -- (\ptwo*\xscale, \Utwo*\yscale);
    \draw[black!45, thin, dashed] (\pthree*\xscale, \Lthree*\yscale) -- (\pthree*\xscale, \Uthree*\yscale);

    % Piecewise-linear upper bound (solid)
    \draw[black, thin]
        (\pone*\xscale, \Uone*\yscale) --
        (\ptwo*\xscale, \Utwo*\yscale) --
        (\pthree*\xscale, \Uthree*\yscale) --
        (\pfour*\xscale, \Ufour*\yscale);

    % Piecewise-linear lower bound (dashed)
    \draw[black, thin, dashed]
        (\pone*\xscale, \Lone*\yscale) --
        (\ptwo*\xscale, \Ltwo*\yscale) --
        (\pthree*\xscale, \Lthree*\yscale) --
        (\pfour*\xscale, \Lfour*\yscale);

    % Function f(x) = sin(x) + 0.3 (thick curve)
    \draw[black, thick, domain=0:3, samples=80]
        plot ({\x*\xscale}, {(sin(\x r) + 0.3)*\yscale});

    % f(x) label (between upper and lower bound labels)
    \node[font=\small] at (3.25*\xscale, 0.425*\yscale) {$f(x)$};

    % Upper bound label
    \node[font=\scriptsize, text=black!50] at (3.25*\xscale, \Ufour*\yscale + 0.15) {$\bar{f}$};
    % Lower bound label
    \node[font=\scriptsize, text=black!50, anchor=north] at (3.25*\xscale, \Lfour*\yscale - 0.15) {$\smash{\underline{f}}$};

    % Vertex markers
    \fill (\pone*\xscale, \Uone*\yscale) circle (1.8pt);
    \fill (\ptwo*\xscale, \Utwo*\yscale) circle (1.8pt);
    \fill (\pthree*\xscale, \Uthree*\yscale) circle (1.8pt);
    \fill (\pfour*\xscale, \Ufour*\yscale) circle (1.8pt);
    \draw[fill=white, thin] (\pone*\xscale, \Lone*\yscale) circle (1.8pt);
    \draw[fill=white, thin] (\ptwo*\xscale, \Ltwo*\yscale) circle (1.8pt);
    \draw[fill=white, thin] (\pthree*\xscale, \Lthree*\yscale) circle (1.8pt);
    \draw[fill=white, thin] (\pfour*\xscale, \Lfour*\yscale) circle (1.8pt);
\end{scope}

\end{tikzpicture}

%% file: sections/examples/polyenc_xample.tex
\begin{example}(Polyhedral Enclosure)
  Consider again the bounding set $\vset_1$ from~\Cref{ex:boundingset}.  Let
  $\Delta_1$ be the triangulation of $\pset_1$ whose 2-simplices are the
  triangles $T_1$, with vertices $\vec{p}_1, \vec{p}_2, \vec{p}_3$, and $T_2$, with vertices
  $\vec{p}_2, \vec{p}_3, \vec{p}_4$.  Then, $\poly(\vset^1_{T_1})$ and $\poly(\vset^1_{T_2})$ are
  the two prisms obtained by slicing the polyhedron formed by $\vset^1$ in half
  along the plane $x_1 + x_2 =0$.  And $\polyenc(\vset^1,\Delta_1)$ is their
  union, which, in this case, is again just the polyhedron formed by $\vset^1$.
\end{example}

%% file: sections/examples/func_xample.tex
\begin{example}(Function Enclosure)
Let $f(x_1,x_2)=x_2\cos(x_1)$. The maximum value of $f$ on
$\dom(\vset^1)$ is $5$ and the minimum value is -5. $\polyenc(\vset^1,\Delta)$ is the origin-centered cube with side length $10$ for every Delaunay point set triangulation $\Delta$ of $\pset_1$.  Thus, $\vset^1$ encloses $f$.
\label{ex:funcenc}
\end{example}

%% file: figures/comp1.tex
\begin{tikzpicture}[
    >=stealth,
    every node/.style={font=\small},
    stepbox/.style={
        draw, rounded corners=3pt, minimum width=2.6cm, minimum height=0.7cm,
        fill=black!6, thin, align=center, font=\footnotesize
    },
    endbox/.style={
        draw, rounded corners=3pt, minimum width=2.6cm, minimum height=0.7cm,
        fill=black!15, thin, align=center, font=\footnotesize
    },
    arr/.style={-{Stealth[length=3pt]}, thin, black!50},
    colhead/.style={font=\small\bfseries, align=center},
]

% === Column positions ===
\def\colA{0}
\def\colB{9.5}
\def\mid{4.75}

% =====================================================
% Shared top: function + decompose
% =====================================================
\node[stepbox, minimum width=4cm] (func) at (\mid, 0)
    {$f(x_3, x_4) = x_4 \cos(x_3)$};
\node[above=2pt, font=\scriptsize, text=black!50] at (func.north)
    {multivariate nonlinear function};

\node[stepbox, minimum width=4cm] (decomp) at (\mid, -1.3)
    {$g_1(x_3) = \cos(x_3)$, \quad $g_2(x_4) = x_4$};
\node[below=0pt, font=\scriptsize, text=black!50] at (decomp.south)
    {decompose};

\draw[arr] (func) -- (decomp);

% Branching arrows
\draw[arr] (decomp.south) -- ++(0, -0.35) -| (\colA, -2.4);
\draw[arr] (decomp.south) -- ++(0, -0.35) -| (\colB, -2.4);

% Column headers (below arrows)
\node[colhead] at (\colA, -2.7) {OVERTVerify};
\node[colhead] at (\colB, -2.7) {OvertPoly};

% Dashed divider
\draw[thin, black!15, dashed] (\mid, -2.2) -- (\mid, -14.0);

% =====================================================
% LEFT: OVERTVerify
% =====================================================

% --- g1 and g2 bounds (compact, side by side, same style as OvertPoly) ---
\begin{scope}[shift={(\colA - 2.2, -4.5)}]
    \draw[thin, ->] (0, 0) -- (2.0, 0) node[right, font=\tiny] {$x_3$};
    \draw[thin, ->] (0, 0) -- (0, 1.5);
    \draw[thick] (0.1, 0.75) -- (0.6, 1.25) -- (1.2, 1.3) -- (1.7, 0.9);
    \draw[thick, dashed] (0.1, 0.5) -- (0.6, 0.95) -- (1.2, 1.0) -- (1.7, 0.6);
    \fill[black!10]
        (0.1, 0.5) -- (0.6, 0.95) -- (1.2, 1.0) -- (1.7, 0.6) --
        (1.7, 0.9) -- (1.2, 1.3) -- (0.6, 1.25) -- (0.1, 0.75) -- cycle;
    \node[font=\tiny] at (0.9, -0.3) {$B_{g_1}$};
\end{scope}

\begin{scope}[shift={(\colA + 0.5, -4.5)}]
    \draw[thin, ->] (0, 0) -- (2.0, 0) node[right, font=\tiny] {$x_4$};
    \draw[thin, ->] (0, 0) -- (0, 1.5);
    \draw[thick] (0.1, 0.25) -- (1.7, 1.25);
    \draw[thick, dashed] (0.1, 0.1) -- (1.7, 0.95);
    \fill[black!10]
        (0.1, 0.1) -- (1.7, 0.95) --
        (1.7, 1.25) -- (0.1, 0.25) -- cycle;
    \node[font=\tiny] at (0.9, -0.3) {$B_{g_2}$};
\end{scope}

% Relational constraint box
\node[stepbox, minimum width=3.0cm] (relbox) at (\colA, -6.8)
    {$y = y_1 \cdot y_2$\\[1pt]
    {\scriptsize implicit relation}};

\draw[arr] (\colA, -5.3) -- (relbox.north);

% MILP
\node[endbox] (milpA) at (\colA, -12.5)
    {MILP\\{\scriptsize (relational constraints)}};
\draw[arr] (relbox) -- (milpA);

% =====================================================
% RIGHT: OvertPoly
% =====================================================

% --- Same g1, g2 bounds (smaller, side by side) ---
\begin{scope}[shift={(\colB - 2.2, -4.5)}]
    \draw[thin, ->] (0, 0) -- (2.0, 0) node[right, font=\tiny] {$x_3$};
    \draw[thin, ->] (0, 0) -- (0, 1.5);
    \draw[thick] (0.1, 0.75) -- (0.6, 1.25) -- (1.2, 1.3) -- (1.7, 0.9);
    \draw[thick, dashed] (0.1, 0.5) -- (0.6, 0.95) -- (1.2, 1.0) -- (1.7, 0.6);
    \fill[black!10]
        (0.1, 0.5) -- (0.6, 0.95) -- (1.2, 1.0) -- (1.7, 0.6) --
        (1.7, 0.9) -- (1.2, 1.3) -- (0.6, 1.25) -- (0.1, 0.75) -- cycle;
    \node[font=\tiny] at (0.9, -0.3) {$B_{g_1}$};
\end{scope}

\begin{scope}[shift={(\colB + 0.5, -4.5)}]
    \draw[thin, ->] (0, 0) -- (2.0, 0) node[right, font=\tiny] {$x_4$};
    \draw[thin, ->] (0, 0) -- (0, 1.5);
    \draw[thick] (0.1, 0.25) -- (1.7, 1.25);
    \draw[thick, dashed] (0.1, 0.1) -- (1.7, 0.95);
    \fill[black!10]
        (0.1, 0.1) -- (1.7, 0.95) --
        (1.7, 1.25) -- (0.1, 0.25) -- cycle;
    \node[font=\tiny] at (0.9, -0.3) {$B_{g_2}$};
\end{scope}

% Arrow: lift + interpolate
\node[stepbox, minimum width=2.8cm, font=\scriptsize] (liftbox) at (\colB, -6.5)
    {Lift to $\mathbb{R}^2$, interpolate};
\draw[arr] (\colB, -5.3) -- (liftbox.north);

% --- 3D composed enclosure (clean upper bound surface) ---
\newcommand{\iso}[3]{({#1 + 0.35*#2}, {#3 + 0.25*#2})}

\begin{scope}[shift={(\colB - 1.8, -10.0)}]
    \def\d{2.0}
    \def\xs{0.85}
    \def\zs{1.5}

    % Upper surface data: U values 0.50, 1.35, 1.40, 0.65

    % Back edge (drawn first, lighter)
    \draw[thin, black!35]
        \iso{0*\xs}{\d}{0.50*\zs} -- \iso{1*\xs}{\d}{1.35*\zs} --
        \iso{2*\xs}{\d}{1.40*\zs} -- \iso{3*\xs}{\d}{0.65*\zs};

    % Surface fill
    \fill[black!12]
        \iso{0*\xs}{0}{0.50*\zs} -- \iso{1*\xs}{0}{1.35*\zs} --
        \iso{2*\xs}{0}{1.40*\zs} -- \iso{3*\xs}{0}{0.65*\zs} --
        \iso{3*\xs}{\d}{0.65*\zs} -- \iso{2*\xs}{\d}{1.40*\zs} --
        \iso{1*\xs}{\d}{1.35*\zs} -- \iso{0*\xs}{\d}{0.50*\zs} -- cycle;

    % Front edge (bold)
    \draw[thick]
        \iso{0*\xs}{0}{0.50*\zs} -- \iso{1*\xs}{0}{1.35*\zs} --
        \iso{2*\xs}{0}{1.40*\zs} -- \iso{3*\xs}{0}{0.65*\zs};

    % Side edges
    \draw[thin] \iso{0*\xs}{0}{0.50*\zs} -- \iso{0*\xs}{\d}{0.50*\zs};
    \draw[thin] \iso{3*\xs}{0}{0.65*\zs} -- \iso{3*\xs}{\d}{0.65*\zs};

    % Internal depth lines at kinks
    \draw[thin, black!25] \iso{1*\xs}{0}{1.35*\zs} -- \iso{1*\xs}{\d}{1.35*\zs};
    \draw[thin, black!25] \iso{2*\xs}{0}{1.40*\zs} -- \iso{2*\xs}{\d}{1.40*\zs};

    % Markers on front edge
    \fill \iso{0*\xs}{0}{0.50*\zs} circle (1.5pt);
    \fill \iso{1*\xs}{0}{1.35*\zs} circle (1.5pt);
    \fill \iso{2*\xs}{0}{1.40*\zs} circle (1.5pt);
    \fill \iso{3*\xs}{0}{0.65*\zs} circle (1.5pt);

    % Ground plane
    \def\gz{-0.3}
    \draw[thin] ({0}, {\gz}) -- ({3*\xs}, {\gz});
    \draw[thin, black!30] ({0}, {\gz}) -- \iso{0}{\d}{\gz};
    \draw[thin, black!30] ({3*\xs}, {\gz}) -- \iso{3*\xs}{\d}{\gz};
    \draw[thin, black!30] \iso{0}{\d}{\gz} -- \iso{3*\xs}{\d}{\gz};

    % Axis labels
    \node[font=\scriptsize] at ({3*\xs + 0.3}, {\gz}) {$x_3$};
    \node[font=\scriptsize, above] at \iso{0}{\d + 0.15}{\gz} {$x_4$};

    % Label
    \node[font=\scriptsize] at ({1.5*\xs + 0.35*\d/2}, {\gz - 0.5}) {$\mathcal{E}(B_f, \Delta)$};
\end{scope}

\draw[arr] (liftbox.south) -- (\colB, -7.7);

% MILP
\node[endbox] (milpB) at (\colB, -12.5)
    {MILP\\{\scriptsize (polyhedral enclosure)}};
\draw[arr] (\colB, -11.5) -- (milpB.north);

\end{tikzpicture}

%% file: figures/lif_interp.tex
\begin{tikzpicture}[
    every node/.style={font=\small},
    >=stealth,
]

% === Layout ===
\def\xscale{1.1}
\def\yscale{2.2}
\def\panelgap{5.2}

% 3D projection: (x,y,z) -> (x + 0.35*y, z + 0.25*y)
\newcommand{\iso}[3]{({#1 + 0.35*#2}, {#3 + 0.25*#2})}

% =====================================================
% Panel (a): Upper bound (1D piecewise linear)
% =====================================================
\begin{scope}[shift={(0,0)}]
    \node[anchor=south, font=\small] at (1.5*\xscale, 1.65*\yscale) {Upper bound $\bar{f}$};
    \node[anchor=south east, font=\small\bfseries] at (-0.35*\xscale, 1.65*\yscale) {(a)};

    % x-axis
    \draw[thin] (-0.3*\xscale, 0) -- (3.4*\xscale, 0);
    \foreach \p/\lab in {0/$p_1$, 1/$p_2$, 2/$p_3$, 3/$p_4$} {
        \draw[thin] (\p*\xscale, -0.06) -- (\p*\xscale, 0.06);
        \node[below, font=\footnotesize] at (\p*\xscale, -0.06) {\lab};
    }

    % U values: 0.50, 1.35, 1.40, 0.65

    % Upper bound line
    \draw[thick]
        (0*\xscale, 0.50*\yscale) -- (1*\xscale, 1.35*\yscale) --
        (2*\xscale, 1.40*\yscale) -- (3*\xscale, 0.65*\yscale);

    % Markers
    \fill (0*\xscale, 0.50*\yscale) circle (2.2pt);
    \fill (1*\xscale, 1.35*\yscale) circle (2.2pt);
    \fill (2*\xscale, 1.40*\yscale) circle (2.2pt);
    \fill (3*\xscale, 0.65*\yscale) circle (2.2pt);

    % Axis label
    \node[font=\footnotesize] at (3.55*\xscale, -0.08) {$x_1$};
\end{scope}

% =====================================================
% Panel (b): Lifted upper bound (3D extruded surface)
% =====================================================
\begin{scope}[shift={(\panelgap, 0)}]
    \node[anchor=south, font=\small] at (1.5*\xscale, 1.65*\yscale) {Lifted $\bar{f}\!\!\uparrow^{S,2}$};
    \node[anchor=south east, font=\small\bfseries] at (-0.35*\xscale, 1.65*\yscale) {(b)};

    \def\depth{2.2}
    \def\xs{1.0}
    \def\zs{2.0}

    % U values: 0.50, 1.35, 1.40, 0.65

    % --- Back edge of surface (drawn first) ---
    \draw[thin, black!40]
        \iso{0*\xs}{\depth}{0.50*\zs} -- \iso{1*\xs}{\depth}{1.35*\zs} --
        \iso{2*\xs}{\depth}{1.40*\zs} -- \iso{3*\xs}{\depth}{0.65*\zs};

    % --- Surface fill (top face) ---
    \fill[black!12]
        \iso{0*\xs}{0}{0.50*\zs} -- \iso{1*\xs}{0}{1.35*\zs} --
        \iso{2*\xs}{0}{1.40*\zs} -- \iso{3*\xs}{0}{0.65*\zs} --
        \iso{3*\xs}{\depth}{0.65*\zs} -- \iso{2*\xs}{\depth}{1.40*\zs} --
        \iso{1*\xs}{\depth}{1.35*\zs} -- \iso{0*\xs}{\depth}{0.50*\zs} -- cycle;

    % --- Front edge of surface ---
    \draw[thick]
        \iso{0*\xs}{0}{0.50*\zs} -- \iso{1*\xs}{0}{1.35*\zs} --
        \iso{2*\xs}{0}{1.40*\zs} -- \iso{3*\xs}{0}{0.65*\zs};

    % --- Depth edges (connecting front to back) ---
    \draw[thin]
        \iso{0*\xs}{0}{0.50*\zs} -- \iso{0*\xs}{\depth}{0.50*\zs};
    \draw[thin]
        \iso{3*\xs}{0}{0.65*\zs} -- \iso{3*\xs}{\depth}{0.65*\zs};

    % --- Internal depth lines at kinks for structure ---
    \draw[thin, black!30]
        \iso{1*\xs}{0}{1.35*\zs} -- \iso{1*\xs}{\depth}{1.35*\zs};
    \draw[thin, black!30]
        \iso{2*\xs}{0}{1.40*\zs} -- \iso{2*\xs}{\depth}{1.40*\zs};

    % Front markers
    \fill \iso{0*\xs}{0}{0.50*\zs} circle (2.2pt);
    \fill \iso{1*\xs}{0}{1.35*\zs} circle (2.2pt);
    \fill \iso{2*\xs}{0}{1.40*\zs} circle (2.2pt);
    \fill \iso{3*\xs}{0}{0.65*\zs} circle (2.2pt);

    % === 2D Ground plane ===
    \def\gz{-0.6}
    % Back edge
    \draw[thin, black!30]
        \iso{0*\xs}{\depth}{\gz} -- \iso{3*\xs}{\depth}{\gz};
    % Right depth edge
    \draw[thin, black!30]
        \iso{3*\xs}{0}{\gz} -- \iso{3*\xs}{\depth}{\gz};
    % Front edge (x1 axis)
    \draw[thin] \iso{-0.2*\xs}{0}{\gz} -- \iso{3.2*\xs}{0}{\gz};
    % Left depth edge (x2 axis)
    \draw[thin] \iso{0*\xs}{0}{\gz} -- \iso{0*\xs}{\depth}{\gz};

    % x1 tick marks on front edge
    \foreach \p/\lab in {0/$p_1$, 1/$p_2$, 2/$p_3$, 3/$p_4$} {
        \draw[thin] ({(\p*\xs)}, {\gz - 0.08}) -- ({(\p*\xs)}, {\gz + 0.08});
        \node[below, font=\footnotesize] at ({(\p*\xs)}, {\gz - 0.1}) {\lab};
    }

    % x1 axis label
    \node[font=\footnotesize] at ({3.4*\xs}, {\gz}) {$x_1$};
    % x2 axis label
    \node[font=\footnotesize, above left=-1pt] at \iso{0*\xs}{\depth+0.2}{\gz} {$x_2$};

    % Dashed drop lines from surface to ground
    \draw[thin, dashed, black!25]
        \iso{0*\xs}{0}{0.50*\zs} -- ({0*\xs}, {\gz});
    \draw[thin, dashed, black!25]
        \iso{1*\xs}{0}{1.35*\zs} -- ({1*\xs}, {\gz});
    \draw[thin, dashed, black!25]
        \iso{2*\xs}{0}{1.40*\zs} -- ({2*\xs}, {\gz});
    \draw[thin, dashed, black!25]
        \iso{3*\xs}{0}{0.65*\zs} -- ({3*\xs}, {\gz});
\end{scope}

% =====================================================
% Panel (c): Upper bound before interpolation
% =====================================================
\begin{scope}[shift={(2*\panelgap, 0)}]
    \node[anchor=south, font=\small] at (1.5*\xscale, 1.65*\yscale) {Upper bound $\bar{f}$};
    \node[anchor=south east, font=\small\bfseries] at (-0.35*\xscale, 1.65*\yscale) {(c)};

    % Coarse: 3 points  U: 0.50, 1.40, 0.65
    \draw[thin] (-0.3*\xscale, 0) -- (3.4*\xscale, 0);
    \foreach \p/\lab in {0/$p_1$, 1.5/$p_2$, 3/$p_3$} {
        \draw[thin] (\p*\xscale, -0.06) -- (\p*\xscale, 0.06);
        \node[below, font=\footnotesize] at (\p*\xscale, -0.06) {\lab};
    }

    \draw[thick]
        (0*\xscale, 0.50*\yscale) -- (1.5*\xscale, 1.40*\yscale) -- (3*\xscale, 0.65*\yscale);

    \fill (0*\xscale, 0.50*\yscale) circle (2.2pt);
    \fill (1.5*\xscale, 1.40*\yscale) circle (2.2pt);
    \fill (3*\xscale, 0.65*\yscale) circle (2.2pt);
\end{scope}

% =====================================================
% Panel (d): After interpolation by q
% =====================================================
\begin{scope}[shift={(3*\panelgap, 0)}]
    \node[anchor=south, font=\small] at (1.5*\xscale, 1.65*\yscale) {Interpolated $\bar{f} +_\Delta \vec{q}$};
    \node[anchor=south east, font=\small\bfseries] at (-0.35*\xscale, 1.65*\yscale) {(d)};

    % q at x=0.75: U(q) = lerp(0.50, 1.40, 0.5) = 0.95
    \def\qx{0.75}
    \def\qU{0.95}

    \draw[thin] (-0.3*\xscale, 0) -- (3.4*\xscale, 0);
    \foreach \p/\lab in {0/$p_1$, 1.5/$p_2$, 3/$p_3$} {
        \draw[thin] (\p*\xscale, -0.06) -- (\p*\xscale, 0.06);
        \node[below, font=\footnotesize] at (\p*\xscale, -0.06) {\lab};
    }
    \draw[thin] (\qx*\xscale, -0.06) -- (\qx*\xscale, 0.06);
    \node[below, font=\footnotesize] at (\qx*\xscale, -0.06) {$\vec{q}$};

    % Dashed ghost of original bound (before interpolation)
    \draw[thin, dashed, black!40]
        (0*\xscale, 0.50*\yscale) -- (1.5*\xscale, 1.40*\yscale);

    % Upper bound line (now with extra kink at q)
    \draw[thick]
        (0*\xscale, 0.50*\yscale) -- (\qx*\xscale, \qU*\yscale) --
        (1.5*\xscale, 1.40*\yscale) -- (3*\xscale, 0.65*\yscale);

    % Original markers
    \fill (0*\xscale, 0.50*\yscale) circle (2.2pt);
    \fill (1.5*\xscale, 1.40*\yscale) circle (2.2pt);
    \fill (3*\xscale, 0.65*\yscale) circle (2.2pt);

    % Interpolated marker (gray)
    \fill[black!60] (\qx*\xscale, \qU*\yscale) circle (2.2pt);
\end{scope}

% =====================================================
% Arrows between panels
% =====================================================
\draw[->, thick, black!50] (3.6*\xscale, 0.75*\yscale) -- (\panelgap - 0.5*\xscale, 0.75*\yscale)
    node[midway, above, font=\scriptsize, text=black!50] {lift};

\draw[->, thick, black!50] (2*\panelgap + 3.6*\xscale, 0.75*\yscale) -- (3*\panelgap - 0.5*\xscale, 0.75*\yscale)
    node[midway, above, font=\scriptsize, text=black!50] {interpolate};

\end{tikzpicture}

%% file: sections/proofs/thm1_proofs.tex
\begin{proof}
  Let $\lift{\vset}{S,n}{\pl,\pu} = \langle n, P', \lift{L}{S,n}{}, \lift{U}{S,n}{} \rangle$, and let $\vec{y}$ be a point in $ \dom(\lift{\vset}{S,n}{\pl,\pu})$.
 
  We need to prove that for any $\Delta$, $L_y\leq \lift{f}{S,n}{}(\vec{y}) \leq U_y$. We consider the validity of the lower bound, as the reasoning for the upper bound is symmetric. 

  Suppose a $\Delta$ exists where $L_y \nleq \lift{f}{S,n}{}(\vec{y})$. Define $\state$ to be the restriction of $\vec{y}$ to $\dom(\vset)$. A similar restriction of the previous inequality yields $L_{\state} \nleq f(\state)$. However, the theorem states that $\vset$ encloses $f$, which is true iff $L_{\state} \leq f(\state)$. This presents a contradiction. 
\end{proof}

%% file: sections/proofs/lem1_proof.tex
\begin{proof}
This is a straightforward consequence of the duality between Delaunay triangulations and Voronoi Diagrams.  Notice that in a grid, the Voronoi cells are hyperrectangles whose vertices are the centers of the grid cells.  Thus, every circumsphere of an $n$-simplex in the Delaunay triangulation must have its center at the center of a grid cell.  Since the circumsphere cannot be larger than the circumsphere of the grid cell by the Delaunay condition and cannot be smaller as it would then touch no points in the grid, it must be exactly the circumsphere of the grid cell.
\end{proof}

%% file: sections/proofs/thm2_proofs.tex
\begin{proof}
    Let $\point$ be an arbitrary point in $\dom(\pset)$, let $\vset' = \vset \interp \vec{q} = \vset' = \langle n,P',L',U'\rangle$, and let $\Delta'$ be a Delaunay triangulation of $P'$. To show that $\vset'$ encloses $f$, we must show that $(\vec{x},f(\vec{x}))\in\polyenc(\vset',\Delta')$ for every $\state \in \dom(\vset')$.
    let $\simp = \simp_{\Delta'}(\vec{x})$, $V_s=\svert(\simp)$, and $\theta = \theta_{\Delta'}(\vec{x})$.
    
    Let $G'$ be the grid cell of $P'$ whose circumsphere coincides with that of $\simp$ (see \Cref{lem:gridsimp}). Let $L_{\state} = \theta \cdot L'(V_s)$, and $U_{\state} = \theta \cdot U'(V_s)$. We consider two cases.  
    
    First, if $G'$ is also a grid cell of $P$, then the simplex $\simp$ must also occur in a Delaunay triangulation of $P$. Recall that $\vset_\simp$ is the restriction of the bounding set $\vset$ to the simplex $\simp$, obtained by restricting $\pset$ to only the points in $V_s$. The restriction of $\vset'$ to $\simp$ is going to be just the same as the restriction of $\vset$ to $\simp$, that is, $\vset_\simp = \vset'_\simp$. 
    So, by substitution, $(\vec{x},f(\vec{x}))\in\poly(\vset_\simp)$ iff $(\vec{x},f(\vec{x}))\in\poly(\vset'_\simp)$.  But we know that $\vset$ encloses $f$, so we must have $(\vec{x},f(\vec{x}))\in\vset_\simp$.  Therefore, $(\vec{x},f(\vec{x}))\in\vset'_\simp$, and so $(\vec{x},f(\vec{x}))\in\polyenc(\vset',\Delta')$.
    
    The more interesting case is when $G'$ is not a grid cell of $P$.  Since $P'$ is a refined version of $P$, we know that $G' \subseteq \conv(G)$ for some grid cell $G$ of $P$.
    For notational convenience, let $I_{\vec{q}}$ be the set of points inserted into $P$ to obtain $P'$.
    
    We wish to show that $L_{\state} \leq f(\state) \leq U_{\state}$. 
    We include only the proof for the upper bound since the proof for the lower bound is symmetric.
    Let $\poly_{G} = \conv(U(G))$ be the polyhedron enclosed by the points in $U(G)$.
    
    Recall that the point $U_{\state}$ is defined as $U_{\state} = \theta \cdot U'(V_s)$. We aim to show that it is also a convex combination of $U(G)$, and thus in $\poly_G$.  First, note that for each $\vec{v}\in V_s \cap G$, we have that $U'(\vec{v})=U(\vec{v}) \in U(G)$.
    
    On the other hand, consider $\vec{v} \in V_s \setminus G$, meaning $\vec{v} \in I_{\vec{q}}$.  In this case, we need to look at the simplex containing $\vec{v}$ in $\Delta$, the triangulation used to interpolate $U$ to get $U'$.  Let $\hat{\simp} = \simp_{\Delta}(\vec{v})$, and let $V=\svert(\hat{\simp})$ be its vertices.  We know, by definition of interpolation, that $U'(\vec{v}) = \theta_{\Delta}(\vec{v})\cdot U(V)$.  Now, if $\vec{v}$ is in the interior of $\conv(G)$, then we must have $V\subseteq G$.  To see why, recall that (by \Cref{lem:gridsimp}) the circumsphere of $\hat{\simp}$ must also be the circumsphere of some grid cell of $P$.  If this grid cell is something other than $G$, then $\vec{v}$ would not be expressible as a convex combination of the points in the grid cell, of which $V$ is a subset, meaning that $\hat{\simp}$ must have the same circumsphere as $G$, and so $V\subseteq G$, and thus $U(V)\subseteq U(G)$
    
    Finally, suppose that $\vec{v}$ is on a face of $\conv(G)$.  In that case, we know that $\theta_{\Delta}(\vec{v})$ is non-zero only for dimensions corresponding to vertices of that face.  Thus, the equation $U'(\vec{v}) = \theta_{\Delta}(\vec{v})\cdot U(V)$ can be rewritten as $U'(\vec{v}) = \theta' \cdot U(V_F)$ for some convex combination vector $\theta'$ and for $V_F \subseteq G$.
    
    Thus, for each $\vec{v}\in V_s$, we have that $U'(\vec{v})$ is a convex combination of points in $U(G)$.  Thus, since $U_{\state}$ is a convex combination of points in $U'(V_s)$, we have that 
    $U_{\state}$ is a convex combination of points in $U(G)$, meaning $U(\state)\in\poly_G$.

    We now claim that $f(\state) \leq U^G_{\state}$, where $U^G_{\state}$ is the point on the lower face of $\poly_G$ at $\state$.
    To see this, note that the face must contain $U(\simp')$ for some $n$-simplex $\simp'$ containing $\state$ and whose vertices are a subset of $G$.  $\simp'$ satisfies the Delaunay condition, since its circumsphere is the same as that of $G$.  It must therefore also be part of some Delaunay triangulation of $P$, and thus, since $\vset$ encloses $f$, we have that $f(\vec{x}) \in \vset_{\simp'}$, and so $f(\state) \leq U^G_{\state} \leq U_{\state}$.
\end{proof}

%% file: figures/bound.tex
\resizebox{0.9\textwidth}{!}{
\begin{tikzpicture}
% Tree Graph
\node[anchor=north] (tree) at (-1cm, 1.5cm) {
    \begin{tikzpicture}[
        grow=down, % Tree grows to the right
        sibling distance=2cm, % Horizontal spacing between sibling nodes
        level distance=1cm, % Vertical spacing between parent and child
        edge from parent/.style={draw, -stealth}, % Arrow style for edges
        every node/.style={circle, draw, minimum size=1.5cm} % Default node style
    ]
        % Root node with its own style
        \node {$\otimes$}
            % First child node with a different style
            child {node {$\mathbf{x}_4$}}
            % Second child node with another style
            child {node {$\cos(\mathbf{x}_3)$}};
    \end{tikzpicture}
};

% Line Plot 1
\begin{axis}[
    name=plot1, % Name for referencing in arrows
    at={(2cm,0)}, % Position
    width=5cm, % Width of the first plot
    height=4cm,
    xlabel={$x_3$},
    ylabel={$f(\mathbf{x}_3)$},
    legend style={at={(1.05,1.2)}, anchor=north west}, % Position legend outside plot
    legend cell align={left} % Align legend entries
]

% Continuous plot
\addplot[color=black, thick, domain=-1:1, samples=1000] {cos(deg(x))};
\addlegendentry{$f(\mathbf{x}_3) = \cos(\mathbf{x}_3)$}

% Scatterplot 1
\addplot[color=red, mark=o, mark size=1pt] coordinates {
 (-1.0, 0.5353023058681398)
 (-0.7162290640736887, 0.7349957172856684)
 (-0.5761808036482412, 0.8335495554233077)
 (-0.23267298160795075, 0.9609732264310585)
 (-0.18889649189051147, 0.9772120444891081)
 (0.18889649189051147, 0.9772120444891081)
 (0.23267298160795, 0.9609732264310588)
 (0.5761808036482412, 0.8335495554233077)
 (0.7162290640736885, 0.7349957172856686)
 (1.0, 0.5353023058681398)
};
\addlegendentry{Lower Bounds}

% Scatterplot 2
\addplot[color=blue, mark=o, mark size=1pt] coordinates {
 (-1.0, 0.5453023058681398)
 (-0.7162290640736887, 0.7840873150595451)
 (-0.5761808036482412, 0.8480683885144896)
 (-0.23267298160795075, 1.005)
 (-0.18889649189051147, 1.005)
 (0.18889649189051147, 1.005)
 (0.23267298160795, 1.005)
 (0.5761808036482412, 0.8480683885144897)
 (0.7162290640736885, 0.7840873150595453)
 (1.0, 0.5453023058681398)
};
\addlegendentry{Upper Bounds}
\end{axis}

% Line Plot 2
\begin{axis}[
    name=plot2, % Name for referencing in arrows
    at={(2cm,-4cm)}, % Position below the first plot
    width=5cm,
    height=4cm,
    xlabel={$x_4$},
    ylabel={$f(\mathbf{x}_4)$},
    legend style={at={(1.05,0.5)}, anchor=north west}, % Position legend outside plot
    legend cell align={left} % Align legend entries
]
% Continuous plot
\addplot[color=black, thick, domain=-1:1, samples=1000] {x};
\addlegendentry{$f(\mathbf{x}_4) = \mathbf{x}_4$}

% Scatterplot 1
\addplot[color=red, mark=o, mark size=1pt] coordinates {
(-1.0, -1.0500000000010001)
(1.0, 0.949999999999)
};
\addlegendentry{Lower Bounds}

% Scatterplot 2
\addplot[color=blue, mark=o, mark size=1pt] coordinates {
(-1.0, -0.949999999999)
(1.0, 1.0500000000010001)
};
\addlegendentry{Upper Bounds}

\end{axis}

\begin{axis}[
    name=surface,
    at={(10.5cm,-1.9cm)}, % Position to the right of line plots
    width=7cm,
    view={30}{15}, % Viewing angle
    xlabel={$\mathbf{x}_3$},
    ylabel={$\mathbf{x}_4$},
    zlabel={$f(\mathbf{x})$},
    clip=false,
    legend style={
        at={(0.5,-0.25)}, % Position legend outside plot area
        anchor=north, % Align legend to this point
    }
]
% Surface defined by specific points
\addplot3[
    surf, mesh/rows=2,
    color=red,
    opacity = 0.5
] coordinates
{
(-1.0, -1.0, -0.5403023058766738)
 (-0.7162290640736887, -1.0, -0.8963621083899191)
 (-0.5761808036482412, -1.0, -0.8566248877968725)
 (-0.23267298160795075, -1.0, -1.1020803207157925)
 (-0.18889649189051147, -1.0, -1.0533638665416554)
 (0.18889649189051147, -1.0, -1.0533638665416554)
 (0.23267298160795, -1.0, -1.1020803207157925)
 (0.5761808036482412, -1.0, -0.8566248877968725)
 (0.7162290640736885, -1.0, -0.8963621083899191)
 (1.0, -1.0, -0.5403023058766738)
 (-1.0, 1.0, 0.5403023058595977)
 (-0.7162290640736887, 1.0, 0.5836293261814305)
 (-0.5761808036482412, 1.0, 0.820474223049743)
 (-0.23267298160795075, 1.0, 0.8298661321463285)
 (-0.18889649189051147, 1.0, 0.9110602224365536)
 (0.18889649189051147, 1.0, 0.9110602224365536)
 (0.23267298160795, 1.0, 0.8298661321463285)
 (0.5761808036482412, 1.0, 0.820474223049743)
 (0.7162290640736885, 1.0, 0.5836293261814305)
 (1.0, 1.0, 0.5403023058595977)
};
\addlegendentry{Lower Bounds}
\addplot3[
    surf,  
    shader=flat,
    color=black,
    domain=-1:1, % Set x domain
    domain y=-1:1, % Set y domain
    samples=50, 
    opacity=0.5
] 
{y*cos(deg(x))};
\addlegendentry{$f(\mathbf{x}) = \mathbf{x}_4\cos(\mathbf{x}_3)$}

Surface defined by specific points
\addplot3[
    surf, mesh/rows=2,
    color=blue,
    opacity = 0.5
] 
coordinates {
(-1.0, -1.0, -0.5403023058595959)
 (-0.7162290640736887, -1.0, -0.62272092395526)
 (-0.5761808036482412, -1.0, -0.8249930561409187)
 (-0.23267298160795075, -1.0, -0.8638929057152325)
 (-0.18889649189051147, -1.0, -0.9288481779474296)
 (0.18889649189051147, -1.0, -0.9288481779474296)
 (0.23267298160795, -1.0, -0.8638929057152325)
 (0.5761808036482412, -1.0, -0.8249930561409187)
 (0.7162290640736885, -1.0, -0.62272092395526)
 (1.0, -1.0, -0.5403023058595959)
 (-1.0, 1.0, 0.5403023058766863)
 (-0.7162290640736887, 1.0, 0.9354537061638268)
 (-0.5761808036482412, 1.0, 0.8611437208880623)
 (-0.23267298160795075, 1.0, 1.1361070942847675)
 (-0.18889649189051147, 1.0, 1.0711518220525704)
 (0.18889649189051147, 1.0, 1.0711518220525704)
 (0.23267298160795, 1.0, 1.1361070942847675)
 (0.5761808036482412, 1.0, 0.8611437208880623)
 (0.7162290640736885, 1.0, 0.9354537061638268)
 (1.0, 1.0, 0.5403023058766863)
};
\addlegendentry{Upper Bounds}
\end{axis}
\end{tikzpicture}
}

%% file: sections/proofs/thm3_proofs.tex
\begin{proof}
Let $\vset^{\fg}=\vset^f\bowtie\vset^g$, let $x\in\dom(\vset^{\fg})$, and let $\Delta$ be a Delaunay triangulation of $\pset$.  Let $\simp=\simp_{\Delta}(\state)$.  We must show that $(\state, f(\vec{x})\bowtie g(\vec{x})) \in \polyenc(\vset^{\fg}_S, \Delta)$.

Let $V=\svert(\simp)$ and $\vec{\theta} = \theta_{\Delta}(\state)$.  Let $L^{\fg}_{\vec{x}} = \vec{\theta}\cdot L^{\fg}(V)$ and $U^{\fg}_{\vec{x}} = \vec{\theta}\cdot U^{\fg}(V)$.  Showing  $(\state, f(\vec{x})\bowtie g(\vec{x})) \in \polyenc(\vset^{\fg}_S, \Delta)$ reduces to showing $L^{\fg}_{\vec{x}} \le f(\vec{x}) \bowtie g(\vec{x}) \le U^{\fg}_{\vec{x}}$.

Let $L^f_{\vec{x}}=\vec{\theta}\cdot L^f(V)$, $U^f_{\vec{x}}=\vec{\theta}\cdot U^f(V)$, $L^g_{\vec{x}}=\vec{\theta}\cdot L^g(V)$, and $U^g_{\vec{x}}=\vec{\theta}\cdot U^g(V)$.  We know that $(\vec{x},f(\vec{x}))\in\poly(\vset^f_{\simp})$ and $(\vec{x},g(\vec{x}))\in\poly(\vset^g_{\simp})$, so we have $L^f_{\vec{x}} \le f(\vec{x}) \le U^f_{\vec{x}}$ and $L^g_{\vec{x}} \le g(\vec{x}) \le U^g_{\vec{x}}$.

Now, if $\bowtie\ = +$, we have $L^{\fg}_{\vec{x}} = \vec{\theta} \cdot L^{\fg}(V) = \vec{\theta} \cdot (L^f(V) + L^g(V)) = \vec{\theta} \cdot L^f(V) + \vec{\theta} \cdot L^g(V) = L^f_{\vec{x}} + L^g_{\vec{x}} \le f(\vec{x}) + g(\vec{x})$.  We can use a similar argument to show that $f(\state)+ g(\state) \le U^{\fg}_{\vec{x}}$.

If $\bowtie\ = -$, we have $L^{\fg}_{\state} = \vec{\theta} \cdot L^{\fg}(V) = \vec{\theta} \cdot (L^f(V) - U^g(V)) = \vec{\theta} \cdot L^f(V) - \vec{\theta} \cdot U^g(V) = L^f_{\state} - U^g_{\state} \le f(\state) - g(\state)$.   We can use a similar argument to show that $f(\state)- g(x) \le U^{\fg}_{\state}$.
\end{proof}

%% file: sections/proofs/thm4_proofs.tex
\begin{proof}
    For $\bowtie\ \in\{\times,\div\}$, let $\vset^{\fg}=\vset^f\bowtie\vset^g$.  Also, let $x\in\dom(\vset^{\fg})$, and let $\Delta$ be a Delaunay triangulation of $\pset$.  Let $\simp=\simp_{\Delta}(\state)$.  We must show that $(\state, f(\vec{x})\bowtie g(\vec{x})) \in \polyenc(\vset^{\fg}_S, \Delta)$. 

    Let $V=\svert(\simp)$ and $\vec{\theta} = \theta_{\Delta}(\state)$.  Let $L^{\fg}_{\vec{x}} = \vec{\theta}\cdot L^{\fg}(V)$ and $U^{\fg}_{\vec{x}} = \vec{\theta}\cdot U^{\fg}(V)$.  Showing  $(\state, f(\vec{x})\bowtie g(\vec{x})) \in \polyenc(\vset^{\fg}_S, \Delta)$ reduces to showing $L^{\fg}_{\vec{x}} \le f(\vec{x}) \bowtie g(\vec{x}) \le U^{\fg}_{\vec{x}}$.

    Let $L^f_{\vec{x}}=\vec{\theta}\cdot L^f(V)$, $U^f_{\vec{x}}=\vec{\theta}\cdot U^f(V)$, $L^g_{\vec{x}}=\vec{\theta}\cdot L^g(V)$, and $U^g_{\vec{x}}=\vec{\theta}\cdot U^g(V)$.  We know that $(\vec{x},f(\vec{x}))\in\poly(\vset^f_{\simp})$ and $(\vec{x},g(\vec{x}))\in\poly(\vset^g_{\simp})$, so we have $L^f_{\vec{x}} \le f(\vec{x}) \le U^f_{\vec{x}}$ and $L^g_{\vec{x}} \le g(\vec{x}) \le U^g_{\vec{x}}$.
    %Now, if $\bowtie\ \in \{\times, \div\}$, we have $L^{\fg}_{\vec{x}} = \theta \cdot L^{\fg}(V)$. 

Now, we know by Lemma~\ref{lem:gridsimp} that $V\subseteq\State$ for some grid cell $\State$.  Thus, since $L^f_{\state} = \vec{\theta}\cdot L^f(V)$, and $L^f(\vec{v}) \geq L^f_{\State}$ for each $\vec{v}\in V$, it follows that $L^f_{\state} \geq L^f_{\State}$, and so $L^f_{\State} \leq f(\state)$.  A similar argument shows that $f(\state)\leq U^f_{\State}$.  And the same analysis for $g$ yields $L^g_\State \leq g(\state) \leq U^g_\State$.  It then follows from interval arithmetic that $L_{\bowtie}(\State) \leq f(\state)\bowtie g(\state) \leq U_{\bowtie}(\State)$.

Finally, since $L^{\fg}_{\state} = \vec{\theta}\cdot L^{\fg}(V)$ and $L^{\fg}(\vec{v}) \leq L_{\bowtie}(\State)$ for each $\vec{v}\in V$, it follows that $L^{\fg}_{\state} \leq L_{\bowtie}(\State)$.  By similar reasoning, $U_{\bowtie}(\State) \leq U^{\fg}_{\state}$.  It follows that
$L^{\fg}_{\vec{x}} \le f(\vec{x}) \bowtie g(\vec{x}) \le U^{\fg}_{\vec{x}}$. 
\end{proof}

%% file: sections/algorithm.tex
In this section, we present \emph{OvertPoly}, an algorithm for verifying nonlinear neural feedback systems using polyhedral enclosures.
Given a discrete-time neural feedback system \dynsys, we compute bounding sets enclosing each transition function $f_i \in \trans$. Using these bounding sets, we
compute an overapproximation $\hat{\traj}^\dynsys(I)$ of the true system trajectory $\traj^\dynsys(I)$ and verify safety by checking that $\hat{\traj}^\dynsys(I)$ satisfies 
a specified reach-avoid property.

To compute the overapproximation $\hat{\traj}^\dynsys(I)$, we encode the bounding set for the transition function $f_i$ as a mixed-integer linear program (MILP). We also represent the neural network controller $\ctrl$ as a MILP. We optimize the resulting MILPs to overapproximate $\nextF^\dynsys(\state_i)$. We repeat this process for each transition function $f_i \in  \trans^\dynsys$, and each time step $t \in [0..T]$ to compute $\hat{\traj}^\dynsys(I)$.
In the following sections, we provide details about each step of the process.
\cref{ssec:bounds} provides details about how the theory of polyhedral enclosures is used to compute bounds for nonlinear functions. \cref{ssec:dyn} describes an efficient combinatorial representation for polyhedral enclosures, and \cref{ssec: reach} describes the procedure for computing $\nextF^\dynsys(\state)_i$. 

\subsection{Computing Bounding Sets for Nonlinear Dynamics}
\label{ssec:bounds}
\begin{algorithm}
\caption{\bound($f$, $G$)}
\begin{algorithmic}[1]
\REQUIRE function $f : \reals^n \to \reals$, $n$-dimensional grid $G$
\ENSURE A bounding set $\vset$ over $G$ that encloses $f$

\IF{$f$ is a constant function}
    \RETURN $\langle n, G, f^G, f^G \rangle$
\ENDIF

\IF{$f$ is univariate}
    \RETURN $\overt(f, G)$ \COMMENT{univariate enclosure}
\ENDIF

\STATE Let $f_1 \bowtie f_2 = f$

\IF{$\bowtie\ \in \{+, -, \times, \div\}$}
   \STATE $\vset \gets \compose(\bowtie, \bound(f_1, G), \bound(f_2, G))$
   \RETURN $\vset$
\ELSE
   \RETURN error \COMMENT{Unsupported operator}
\ENDIF
\end{algorithmic}
\label{alg:bound}
\end{algorithm}

Algorithm~\ref{alg:bound} provides an overview of the process for computing bounding sets. The inputs to this algorithm are a function $ f $ and a grid $G$ (initially one whose domain is $\init$).
The output is a bounding set $ \vset $ enclosing $f$.
\Cref{alg:bound} computes bounding sets by recursively visiting the components of $f$.  For constant functions, an exact bounding set can be constructed from the function itself.  For univariate functions, we use the method based on OVERT described earlier.  In the general case, functions that consist of an arithmetic operator applied to two functions with enclosing bounding sets can be bounded by applying the bounding set composition technique described above.

\subsection{Generating an Efficient Optimization Model}\label{ssec:dyn}
Overapproximating $\nextF^\dynsys$ requires an efficient combinatorial representation for polyhedral enclosures. Let $\vset=\langle n,\pset,L,U\rangle$ be a bounding set. We employ the aggregated convex combination method ~\citep{geissler2011using}. Let $\Delta$ be a Delaunay triangulation of $\pset$, with $\{S_1,\dots,S_m\}$ denoting the $n$-simplices in $\Delta$. We would like to represent an arbitrary $\state \in \dom(\pset)$ as the convex combination of the vertices of some simplex $S_i$.

We define a binary vector $\vec{b}$ of size $m$ to encode the \emph{active} simplex in $\Delta$. In other words, $\state \in S_i$ iff $b_i = 1$ for every $i \in [m]$. Let $\pset = \{\vec{p}_1,\dots,\vec{p}_d\}$.  We define a convex combination variable $\vec{\lambda}$ that ranges over all points in $\pset$ (i.e., $\vec{\lambda} \in \reals^d$). We can then define an integer program $\model$ for $\vset$.

\begin{subequations}
\begin{gather}
    \sum_{j = 1}^{d} \vec{\lambda}_j = 1, \;\; \vec{\lambda} \vgeq \vec{0}, \;\;\; \sum_{i = 1}^m \vec{b}_i \leq 1,  \;\;\; \vec{b} \in \{0,1\}^m \label{eq:cc1}\\
    \vec{\lambda}_j \leq \sum_{\{i|\vec{p}_j \in \svert(S_i)\}} \vec{b}_i \;\;\; \text{for} \;\; j \in [d] \label{eq:cc2} \\
    \state = \sum_{j = 1}^{d} \vec{\lambda}_j \cdot\vec{p}_j \label{eq:cc3}\\
    \overline{y} = \sum_{j = 1}^{d} \vec{\lambda}_j \cdot U(\vec{p}_j), \;\;\;
    \underline{y} = \sum_{j = 1}^{d} \vec{\lambda}_j \cdot L(\vec{p}_j) \label{eq:cc4}\\ 
    \underline{y} \leq y \leq \overline{y} \label{eq:cc5}
\end{gather}
\label{eq:convexComb}
\end{subequations}

\Cref{eq:cc1} defines the auxiliary variables $\vec{\lambda} $ and $ \vec{b} $ and enforces that only one binary variable can be active at a time. \Cref{eq:cc2} defines the so-called SOS-2 constraint~\citep{beale1976global}, which requires that only points from simplices adjacent to the active simplex can have nonzero convex combination weights. \Cref{eq:cc3} then defines the state variable $\vec{x}$.  To complete the program, \Cref{eq:cc4} defines the polyhedral enclosure induced by $\vset$ and $\Delta$, and \Cref{eq:cc5} states that the output variable $y$ must be inside the enclosure. 

We construct a separate mixed integer program for each transition function $ f_i $ in $\trans$. This yields a set of programs $ \{ \mathcal{M}_1, \dots, \mathcal{M}_n\} $ with variables $\{y_1,\dots,y_n\}$ representing the outputs of the transition functions. 
\Cref{fig:tri} illustrates these constraints for the Unicycle example.
\begin{figure}[ht!]
    \centering
    {\input{figures/tri}}
    \caption{\footnotesize A visualization of Equations \ref{eq:cc1} -- \ref{eq:cc2} for a lower dimensional projection of our running example. \Cref{eq:cc3} is implicitly defined here as well because $ \vec{x} $ is a sample from the convex hull of these points.}
    \label{fig:tri}
\end{figure}
We also construct an integer program $\model_0$ to represent $\ctrl$, the neural network controller. The construction uses standard neural network encoding techniques, as described in
\Cref{sec:MipVerify}.

\subsection{Computing forward reachable sets}
\modded{Modified figures for encoding and dependency graph}
\label{ssec: reach}
To compute $\traj^\dynsys(\State_0)_T$ (i.e. the system state at time step $T$ starting from initial state $\State_0$), we must compute $\State_{t+1} =  \nextF^\dynsys(\State_{t})$ for every $t \in [0..T-1]$. This process is called \emph{forward reachability analysis} \citep{bansal2017hamilton}. Since we deal with nonlinear systems, an exact computation of $\nextState$ is intractable. 
Instead, we focus on computing an overapproximation $\hnextState$ such that $\nextState \subseteq \hnextState$.
To do this, assume $\hat{\State}_t$ is an overapproximation of the set of states at time $t$.  For each $i\in[n]$, we solve the following integer program.
\begin{subequations}
\begin{align}
    \min_{\state \in \hat{\State}_t} \state_i + (y_i + \ctrl(\state)_i + \epsilon) \cdot \delta \\
    \quad \model_0, \dots, \model_n,\\
    \quad\epsilon \in \noise
    \label{eq:LowerBound}
\end{align}
\label{eq:Reach}
\end{subequations}
\noindent
For each $i\in [n]$, solving the problem in \cref{eq:Reach} yields a lower bound $l_i$ for $\state_i$ in the next time step.  Similarly, replacing $\min$ with $\max$ and solving yields an upper bound $u_i$.  We can then define $\hnextState = [l_1,u_1] \times \ldots \times [l_n,u_n]$. 
Constructing $\hat{\traj}^\dynsys(I)$ sequentially by repeating this process for each $t \in [0..T-1]$ is called \emph{concrete} reachability analysis \citep{sidrane2022overt}.

This approach typically introduces more uncertainty each time we solve \cref{eq:Reach}. 
Therefore, sequential application of \cref{eq:Reach} may lead to a coarse overapproximation of the reachable sets, a situation sometimes referred to as \emph{excess conservatism}. To address this, we can instead represent two steps at once.  For $i\in[0..m]$, let $\model'_i$ be the same as $\model_i$, except that each variable $v$ appearing in $\model_i$ is replaced by $v'$ in $\model'_i$.  To solve the two-step problem, we then solve:
\begin{subequations}
\begin{align}
    \min_{\state \in \hat{\State}_t} \state'_i + (y'_i + \ctrl'(\state')_i + \epsilon') \cdot \delta \\
    \quad \model_0, \dots, \model_n\\
    \text{for each }i\in[m]
\begin{cases}
\state'_i = \state_i + (y_i + \ctrl(\state)_i + \epsilon_i),\\ \epsilon_i\in\noise
\end{cases}\\
    \quad \model'_0, \dots, \model'_n\\
    \quad \epsilon' \in \noise
    \label{eq:SymLowerBound}
\end{align}
\label{eq:SymReach}
\end{subequations}
\noindent
This yields lower bounds for $\state_{t+2}$ (as before, replacing $\min$ by $\max$ produces upper bounds).
This process can be repeated to generalize from a two-step analysis to a $k$-step analysis for arbitrary $k$. Constructing $\hat{\traj}^\dynsys(I)$ this way is called \emph{symbolic} reachability analysis \citep{sidrane2022overt}.

\begin{figure}[htpb]
    \center
    \resizebox{\textwidth}{!}{\input{figures/depgraph}}
    \caption{\footnotesize Dependency graph for Unicycle example. The light-grey vertices represent state variables, while the solid grey vertex represents the neural network controller. The arrows on the edges denote dependency direction, with the model at the destination depending on variables updated at the source model. Arrows between bounding boxes denote dependencies across time steps. Note that the state dependencies are identical across time steps}
    \label{fig:SDG}
\end{figure} 

The key challenge with symbolic reachability is scalability.  This can be partly addressed by tracking dependencies with a dependency graph.  The models $\model_0,\dots,\model_m$ are the vertices, and there is an edge between $\model_i$ and $\model_j$
if the output of $\model_i$ depends on the previous output value of $\model_j$ or vice versa.  We use dependency graphs to prune the set of models included in each integer linear program.
\Cref{fig:SDG} illustrates the variable dependency graph induced by the combination of state and temporal dependency in our running example.
% \begin{figure*}[htpb]
%     \center
%     \input{figures/symReach}
%     \caption{\footnotesize Symbolic dependency graph for the Unicycle car model. The state dependency graphs are identical at each time step (since the transition functions are time invariant), but there are edges connecting models across time steps.}
%     \label{fig:sym}
% \end{figure*}
% We compute forward reachable sets by solving the integer linear programs in \cref{eq:Reach,eq:SymReach}. We use the variable dependency graph structure to encode state dependency in the case of concrete reachability, and state/temporal dependency in the case of symbolic reachability. 

%% file: figures/tri.tex
\begin{tikzpicture}[
    every node/.style={font=\small},
    >=stealth,
    gridpt/.style={circle, draw, inner sep=0pt, minimum size=6pt, fill=white, thin},
    gridptfill/.style={circle, draw, inner sep=0pt, minimum size=6pt, fill=black!70, thin},
]

% === 2D Grid of points ===
% A 5x2 grid (10 points), triangulated into 8 triangles
% This is a "lower dimensional projection" of the unicycle example

\def\dx{1.6}  % horizontal spacing
\def\dy{1.6}  % vertical spacing

% Grid positions: (col, row) for col=0..4, row=0..1
% Lambda indexing: bottom row λ1..λ5, top row λ6..λ10

% --- Draw triangulation (fill simplices first, behind everything) ---

% Each grid square is split into 2 triangles along the diagonal
% Convention: bottom-left to top-right diagonal

% Simplex fills (alternating shading for clarity)
% Square (0,0)-(1,0)-(1,1)-(0,1) -> two triangles
\fill[black!6]  (0*\dx, 0*\dy) -- (1*\dx, 0*\dy) -- (0*\dx, 1*\dy) -- cycle;
\fill[black!12] (1*\dx, 0*\dy) -- (1*\dx, 1*\dy) -- (0*\dx, 1*\dy) -- cycle;

\fill[black!6]  (1*\dx, 0*\dy) -- (2*\dx, 0*\dy) -- (1*\dx, 1*\dy) -- cycle;
\fill[black!12] (2*\dx, 0*\dy) -- (2*\dx, 1*\dy) -- (1*\dx, 1*\dy) -- cycle;

\fill[black!6]  (2*\dx, 0*\dy) -- (3*\dx, 0*\dy) -- (2*\dx, 1*\dy) -- cycle;
\fill[black!12] (3*\dx, 0*\dy) -- (3*\dx, 1*\dy) -- (2*\dx, 1*\dy) -- cycle;

\fill[black!6]  (3*\dx, 0*\dy) -- (4*\dx, 0*\dy) -- (3*\dx, 1*\dy) -- cycle;
\fill[black!12] (4*\dx, 0*\dy) -- (4*\dx, 1*\dy) -- (3*\dx, 1*\dy) -- cycle;

% --- Triangulation edges ---
% Horizontal edges
\draw[thin, black!40] (0,0) -- (4*\dx, 0);
\draw[thin, black!40] (0, \dy) -- (4*\dx, \dy);
% Vertical edges
\foreach \c in {0,...,4} {
    \draw[thin, black!40] (\c*\dx, 0) -- (\c*\dx, \dy);
}
% Diagonal edges
\foreach \c in {0,...,3} {
    \draw[thin, black!40] (\c*\dx, \dy) -- ({(\c+1)*\dx}, 0);
}

% --- Highlight one active simplex ---
% Highlight simplex S_k (triangle: (1,0)-(2,0)-(1,1)) to show SOS-2
\draw[thick, black] (1*\dx, 0*\dy) -- (2*\dx, 0*\dy) -- (1*\dx, 1*\dy) -- cycle;

% --- Grid points with lambda labels ---
% Bottom row: λ1 .. λ5
\foreach \c/\idx in {0/1, 1/2, 2/3, 3/4, 4/5} {
    \node[gridpt] (L\idx) at (\c*\dx, 0) {};
    \node[below=4pt, font=\footnotesize] at (L\idx) {$\lambda_{\idx}$};
}
% Top row: λ6 .. λ10
\foreach \c/\idx in {0/6, 1/7, 2/8, 3/9, 4/10} {
    \node[gridpt] (L\idx) at ({(\c)*\dx}, \dy) {};
    \node[above=4pt, font=\footnotesize] at (L\idx) {$\lambda_{\idx}$};
}

% Highlight the active vertices of the selected simplex
\node[gridptfill] at (1*\dx, 0) {};
\node[gridptfill] at (2*\dx, 0) {};
\node[gridptfill] at (1*\dx, \dy) {};

% --- Simplex label ---
\node[font=\footnotesize, text=black!50] at (1.3*\dx, 0.45*\dy) {$S_k$};

% --- Binary variable annotation ---
\draw[->, thin, black!50] (1.3*\dx, 0.25*\dy) -- (1.3*\dx, -0.6*\dy);
\node[font=\footnotesize, text=black!50, below] at (1.3*\dx, -0.65*\dy)
    {$b_k = 1$};

% --- Right side: constraint summary ---
\node[anchor=west, font=\footnotesize, text=black!60, align=left] at (4*\dx + 0.8, 0.85*\dy) {%
$\displaystyle\sum_{j=1}^{d} \lambda_j = 1,\quad \lambda \succeq 0$};

\node[anchor=west, font=\footnotesize, text=black!60, align=left] at (4*\dx + 0.8, 0.35*\dy) {%
$\displaystyle\lambda_j \leq \!\!\sum_{\{i\,|\,p_j \in \mathrm{vert}(S_i)\}}\!\! b_i$};

\node[anchor=west, font=\footnotesize, text=black!60, align=left] at (4*\dx + 0.8, -0.15*\dy) {%
$\displaystyle\sum_{i=1}^{m} b_i \leq 1,\quad b \in \{0,1\}^m$};

% --- Legend ---
\node[gridptfill, minimum size=5pt] (leg1) at (4*\dx + 1.0, -0.65*\dy) {};
\node[right=2pt, font=\scriptsize] at (leg1) {$\lambda_j > 0$ permitted};
\node[gridpt, minimum size=5pt] (leg2) at (4*\dx + 1.0, -0.95*\dy) {};
\node[right=2pt, font=\scriptsize] at (leg2) {$\lambda_j = 0$ enforced};

% --- x-axis labels ---
\node[font=\footnotesize, text=black!50] at (4.3*\dx, -0.35*\dy) {$x_3$};
\node[font=\footnotesize, text=black!50, rotate=90] at (-0.35*\dx, 0.5*\dy) {$x_4$};

\end{tikzpicture}

%% file: figures/depgraph.tex
\begin{tikzpicture}[
    >=stealth,
    every node/.style={font=\small},
    statenode/.style={
        draw, rounded corners=2pt, minimum width=1.2cm, minimum height=0.65cm,
        fill=black!6, thin, align=center, font=\scriptsize
    },
    ctrlnode/.style={
        draw, rounded corners=2pt, minimum width=1.2cm, minimum height=0.65cm,
        fill=black!15, thin, align=center, font=\scriptsize
    },
    depedge/.style={thin, black!40},
    sharedge/.style={thin, black!25, dashed},
    tempedge/.style={thin, black!60, -{Stealth[length=3pt]}},
    timestep/.style={draw=black!20, rounded corners=6pt, inner xsep=8pt, inner ysep=14pt, fill=black!3},
]

\def\tgap{7.0}   % gap between time steps
\def\vgap{1.8}   % vertical gap (controller above states)
\def\hgap{1.45}   % horizontal gap between state nodes

% === Time step 1 ===
\begin{scope}[shift={(0,0)}]
    % State nodes
    \node[statenode] (M11) at (0, 0) {$M^1_1$\\[-1pt]$x_4\cos x_3$};
    \node[statenode] (M12) at (\hgap, 0) {$M^1_2$\\[-1pt]$x_4\sin x_3$};
    \node[statenode] (M13) at (2*\hgap, 0) {$M^1_3$\\[-1pt]$u_2$};
    \node[statenode] (M14) at (3*\hgap, 0) {$M^1_4$\\[-1pt]$u_1\!+\!w$};
    % Controller
    \node[ctrlnode] (M10) at (1.5*\hgap, \vgap) {$M^1_0$};
    % Edges
    \draw[depedge, ->] (M11) -- (M10);
    \draw[depedge, ->] (M12) -- (M10);
    \draw[depedge, <->] (M10) -- (M13);
    \draw[depedge, <->] (M10) -- (M14);
    % Dependency arcs (curved below): M1,M2 depend on M3,M4
    \draw[thin, black!40, ->, bend left=30] (M13.south) to (M11.south);
    \draw[thin, black!40, ->, bend left=25] (M14.south) to (M11.south);
    \draw[thin, black!40, ->, bend left=25] (M13.south) to (M12.south);
    \draw[thin, black!40, ->, bend left=30] (M14.south) to (M12.south);
    % Background group
    \begin{pgfonlayer}{background}
        \node[timestep, fit=(M11)(M14)(M10), label={[font=\footnotesize, text=black!50]above:$t=1$}] {};
    \end{pgfonlayer}
\end{scope}

% === Time step 2 ===
\begin{scope}[shift={(\tgap,0)}]
    \node[statenode] (M21) at (0, 0) {$M^2_1$\\[-1pt]$x_4\cos x_3$};
    \node[statenode] (M22) at (\hgap, 0) {$M^2_2$\\[-1pt]$x_4\sin x_3$};
    \node[statenode] (M23) at (2*\hgap, 0) {$M^2_3$\\[-1pt]$u_2$};
    \node[statenode] (M24) at (3*\hgap, 0) {$M^2_4$\\[-1pt]$u_1\!+\!w$};
    \node[ctrlnode] (M20) at (1.5*\hgap, \vgap) {$M^2_0$};
    \draw[depedge, ->] (M21) -- (M20);
    \draw[depedge, ->] (M22) -- (M20);
    \draw[depedge, <->] (M20) -- (M23);
    \draw[depedge, <->] (M20) -- (M24);
    \draw[thin, black!40, ->, bend left=30] (M23.south) to (M21.south);
    \draw[thin, black!40, ->, bend left=25] (M24.south) to (M21.south);
    \draw[thin, black!40, ->, bend left=25] (M23.south) to (M22.south);
    \draw[thin, black!40, ->, bend left=30] (M24.south) to (M22.south);
    \begin{pgfonlayer}{background}
        \node[timestep, fit=(M21)(M24)(M20), label={[font=\footnotesize, text=black!50]above:$t=2$}] {};
    \end{pgfonlayer}
\end{scope}

% === Dots between t=2 and t=T ===
\node[font=\large, text=black!40] at (2*\tgap - 1.5, 0.9) {$\cdots$};

% === Time step T ===
\begin{scope}[shift={(2*\tgap,0)}]
    \node[statenode] (MT1) at (0, 0) {$M^t_1$\\[-1pt]$x_4\cos x_3$};
    \node[statenode] (MT2) at (\hgap, 0) {$M^t_2$\\[-1pt]$x_4\sin x_3$};
    \node[statenode] (MT3) at (2*\hgap, 0) {$M^t_3$\\[-1pt]$u_2$};
    \node[statenode] (MT4) at (3*\hgap, 0) {$M^t_4$\\[-1pt]$u_1\!+\!w$};
    \node[ctrlnode] (MT0) at (1.5*\hgap, \vgap) {$M^t_0$};
    \draw[depedge, ->] (MT1) -- (MT0);
    \draw[depedge, ->] (MT2) -- (MT0);
    \draw[depedge, <->] (MT0) -- (MT3);
    \draw[depedge, <->] (MT0) -- (MT4);
    \draw[thin, black!40, ->, bend left=30] (MT3.south) to (MT1.south);
    \draw[thin, black!40, ->, bend left=25] (MT4.south) to (MT1.south);
    \draw[thin, black!40, ->, bend left=25] (MT3.south) to (MT2.south);
    \draw[thin, black!40, ->, bend left=30] (MT4.south) to (MT2.south);
    \begin{pgfonlayer}{background}
        \node[timestep, fit=(MT1)(MT4)(MT0), label={[font=\footnotesize, text=black!50]above:$t=T$}] {};
    \end{pgfonlayer}
\end{scope}

% === Temporal arrows (single arrow between each box) ===
% t=1 -> t=2
\draw[-{Stealth[length=4pt]}, thick, black!50]
    (3*\hgap + 0.55, 0.9) -- (\tgap - 0.55, 0.9);
% t=2 -> dots
\draw[-{Stealth[length=4pt]}, thick, black!50]
    (\tgap + 3*\hgap + 0.55, 0.9) -- (2*\tgap - 1.8, 0.9);
% dots -> t=T
\draw[-{Stealth[length=4pt]}, thick, black!50]
    (2*\tgap - 1.1, 0.9) -- (2*\tgap - 0.55, 0.9);

% % === Legend ===
% \node[anchor=west, font=\scriptsize, text=black!60] at (2, -1.3)
%     {Solid: controller dependency \quad Curved: state dependency};

\end{tikzpicture}

%% file: sections/optimizations/depgraph.tex
We include the ILP $\model_j$ in the optimization problem for a state variable $\state_i$ if $\model_i$ and $\model_j$ share an edge in the state dependency graph of the problem. 
However, when solving a symbolic reachability problem over $k$ steps, the model at step $k$ ($\model_i^k$) depends not only on its immediate neighbors in the graph, but also on its own history, specifically, on $\model_i^{k-1}, \model_i^{k-2}, \ldots, \model_i^{0}$. This results in the extended dependency structure illustrated in \cref{fig:SDG}.
To manage these dependencies efficiently, we use the Plasmo framework~\cite{Jalving2022}, which prunes irrelevant ILPs based on the structure of the provided graph.

%% file: sections/optimizations/mccormick.tex
The nonlinear composition introduced in Definition \ref{def:nln_comp} is sound but may be conservative in practice. In cases where the composition yields a bilinear term of the form $z = x \cdot y$, we can tighten the bounds using the McCormick envelope~\cite{hijazi2019perspective}.
\begin{align*}
z &\geq x^\ell y + y^\ell x - x^\ell y^\ell \\
z &\geq x^u y + y^u x - x^u y^u \\
z &\leq x^\ell y + y^u x - x^\ell y^u \\
z &\leq x^u y + y^\ell x - x^u y^\ell
\end{align*}
\noindent
where $x \in [x^\ell, x^u]$ and $y \in [y^\ell, y^u]$.

%% file: sections/optimizations/crown.tex
The ReLU network encoding described in \cref{sec:MipVerify} is sensitive to the quality of pre-activation bounds used to construct the MILP. The original paper proposed interval arithmetic or LP relaxations for pre-activation bound computation, but used interval arithmetic for computational efficiency. Authors in OVERTVerify~\citep{sidrane2022overt} used a version of MIPVerify that replaced interval arithmetic with the tighter MaxSens~\citep{xiang2018output} encoding. We explore further tightening the pre-activation bounds by using the CROWN solver~\citep{zhang2018efficient}. This approach improves on MaxSens by incorprating a backward pass, further tightening pre-activation bounds for deep networks. 

%% file: sections/optimizations/ideal_relu.tex
The univariate ReLU encoding defined in \cref{eq:mipverify} is \emph{ideal}, meaning that all vertices of its LP relaxation are integral. In practice, however, one typically encounters multivariate affine functions (e.g., pre-activation layers) $f : [L,U] \to \reals$, where $n \in \nats$, $L,U \in \reals^n$, and $L \vleq x \vleq U$. To encode such functions, we instead employ the \emph{Big-M} formulation.
\begin{subequations}
\begin{gather}
    y \leq f(x) - M^-(f;[L,U]) \cdot (1-z), \quad y \geq f(x), \quad y \leq M^+(f;[L,U])\cdot z, \\ x \in [L,U], \quad y\in \reals_+, \quad z \in \{0,1\}
\end{gather}
\end{subequations}
Unfortunately, the \emph{ideal} property does not extend to the multivariate setting. \citet{anderson2020strong} strengthen this formulation by introducing an ideal encoding for multivariate ReLU functions. We define the following element-wise coefficients:
\begin{gather}
    \hat{L_i} = \begin{cases}
        L_i \quad \text{if } w_i \geq 0 \\
        U_i \quad \text{if } w_i < 0
    \end{cases} \quad \text{and} \quad 
    \hat{U_i} = \begin{cases}
        U_i \quad \text{if } w_i \geq 0 \\
        L_i \quad \text{if } w_i < 0
    \end{cases}
\end{gather}
Then for some function $f(x) = w \cdot x + b$ defined over $[L,U]$, we have the following
\begin{subequations}
\begin{gather}
    y \leq \sum_{i \in I} w_i(x_i - \hat{L}(1-z)) + (b + \sum_{i \notin I}w_i\hat{U}_i) \cdot z \quad \forall I \subseteq 2^{[n]}, \label{eq:relusharp1}\\ y \geq w \cdot x + b, \quad x \in [L,U], \quad y \in \reals_+, \quad z \in \{0,1\}
\end{gather}
\end{subequations}
\Cref{eq:relusharp1} depends on the index set $I$, whose size grows exponentially with the input dimension, rendering this formulation impractical as a MILP encoding. Nonetheless, \citet{anderson2020strong} observe that these results can be leveraged to strengthen the formulation on an as-needed basis.

%% file: sections/optimizations/dcc.tex
The convex combination encoding defined in \cref{eq:convexComb} introduces one binary variable for each simplex in the triangulation, which becomes problematic in higher dimensions (e.g., $n > 3$). This can be improved by exploiting the fact that only $\lceil \log_2(m) \rceil$ binary variables are required to encode $m$ distinct states, where each state corresponds to a single active simplex among the $m$ simplices~\citep{geissler2011using}. 
Let $\simp := \{\simp_1,\ldots,\simp_m\}$ denote the set of $n$-simplices in the triangulation $\Delta$, and let $c : \simp \to \{0,1\}^{\lceil \log_2(m) \rceil}$ be an injective encoding function. We then define the following alternative integer program $\mathcal{M}$:
\begin{subequations}
\begin{gather}
    \sum_{i=1}^m \sum_{j=0}^n \lambda_j^{\simp_i} = 1, \quad \vec{\lambda} \vgeq \vec{0}, \quad \vec{b} \in \{0,1\}^{\rceil log_2(m) \rceil} \label{eq:dcc1}\\
    \sum_{i = 1}^m \sum_{j=0}^n c(\simp_i)_k\lambda_j^{\simp_i} \leq b_k \quad \text{for } k = 1,\ldots,\lceil log_2(m) \rceil \label{eq:dcc2}\\
    \sum_{i = 1}^m \sum_{j=0}^n (1 - c(\simp_i)_k)\lambda_j^{\simp_i} \leq 1-b_k \quad \text{for } k = 1,\ldots,\lceil log_2(m) \rceil \label{eq:dcc3}\\
    \state = \sum_{i=1}^m \sum_{j=0}^n \lambda_j^{\simp_i}p_j^{\simp_i} \label{eq:dcc4}\\ 
    \overline{y} = \sum_{i=1}^m \sum_{j=0}^n \lambda_j^{\simp_i}U(p_j^{\simp_i}) \quad \underline{y} = \sum_{i=1}^m \sum_{j=0}^n \lambda_j^{\simp_i}U(p_j^{\simp_i}) \label{eq:dcc5}\\
    \underline{y} \leq y \leq \overline{y} \label{eq:dcc6}
\end{gather} 
\end{subequations}
\Cref{eq:dcc1} introduces auxiliary variables for the \emph{disaggregated} convex combination encoding, in which each simplex $\simp_i$ is associated with $n+1$ distinct convex combination variables. In contrast, \cref{eq:cc1} defines an \emph{aggregated} formulation, where the $\lambda$ variables are shared across simplices. The disaggregated encoding trades a larger number of continuous variables for a reduced number of binary variables, and yields a tighter LP relaxation than the aggregated formulation~\cite{geissler2011using}. 
\Cref{eq:dcc2,eq:dcc3} impose bit-level constraints via the injective function, ensuring that at most one simplex is active at any time. The remaining constraints follow the standard convex combination encoding.

%% file: sections/implementation.tex
We exploit the meta-programming features of the Julia programming language for expression parsing and symbolic analysis. We use the Symbolics.jl~\citep{10.1145/3511528.3511535} package for analytical derivatives, IntervalRootfinding.jl~\citep{IntervalRootFinding} for sound root finding, and OVERT.jl~\citep{sidrane2022overt} for univariate bound computation. We combine methods from OVERTVerify.jl and NeuralVerification.jl~\citep{liu2020algorithmsverifyingdeepneural} for our local implementation of MIPVerify. We use the Plasmo framework ~\citep{jalving2019graph,Jalving2022} to solve MILPs on graphs, with Gurobi as a backend solver~\citep{gurobi}.

%% file: sections/eval.tex
We begin by comparing our tool against state-of-the-art combinatorial and propagation-based approaches to reachability analysis. We first describe the benchmark set and then present results for an unoptimized version of our algorithm. Next, we examine the limitations of our approach through a more detailed comparison of the combinatorial methods. Finally, we conclude with an ablation study.

\subsection{Benchmarks}
We evaluate all algorithms on a subset of the benchmarks used in the annual ARCH competition~\citep{lopez2023arch}. Since we only support ReLU networks, we restrict ourselves to benchmarks that use ReLU networks. 
We describe pertinent details about some of these benchmarks
below.
\input{sections/benchmarks}
\subsection{Comparing to state-of-the-art approaches}
We assess the quality of our contribution by comparing it to a discrete-time implementation of CORA~\citep{kochdumper2023constrained}, the state-of-the-art propagation based tool, and OVERTVerify~\citep{sidrane2022overt}, the state-of-the-art combinatorial tool. 
We use the Polynomial Zonotope abstraction \citep{kochdumper2023constrained} in CORA. We note that the performance of the CORA algorithm depends on the choice of hyperparameters, and we report results for the lowest-precision settings needed to verify a given property.
\subsubsection{Experimental Setup}
Computation times were obtained by running all tools on an AMD Ryzen 9 7950x processor. For the pendulum, ACC, and TORA benchmarks, we used concrete reachability and~\Cref{alg:bound} as shown.
For the Unicycle benchmark, which is more challenging, we used five iterations of symbolic reachability with $k=10$ (this technique is called \emph{hybrid-symbolic reachability} in~\citep{sidrane2022overt}) for both OvertPoly and OVERTVerify.  For OvertPoly, we also replaced the nonlinear composition described above with a more precise method specialized for low dimensions based on McCormick envelopes \citep{hijazi2019perspective}.

\subsubsection{Baseline Results}
We define our baseline algorithm as one that employs the aggregated convex combination encoding, dependency graphs, and McCormick envelopes. To ensure a fair comparison, we exclude additional network-level optimizations. Baseline results are reported in \Cref{tab:results}. 
For each benchmark, we report the time required to either verify or falsify the specified property, along with the volume of the reachable set at time $\horizon$. Since CORA does not support volume computation for polynomial zonotopes, we approximate the volume via grid-based sampling. 
Despite employing a combinatorial approach, the resulting computation times from our method are comparable to those of state-of-the-art propagation-based methods.
\begin{table*}[ht!]
\centering
\resizebox{0.9\columnwidth}{!}{
    \begin{tabular}{lcccccc}
        \toprule
        & \multicolumn{2}{c}{OvertPoly (\ding{72})} & \multicolumn{2}{c}{OVERTVerify} & \multicolumn{2}{c}{CORA} \\ 
        \cmidrule(lr){2-3} \cmidrule(lr){4-5} \cmidrule(lr){6-7}
        & Time (s) & Volume & Time (s) & Volume &  Time (s) & Volume \\ 
        \midrule
        Pend. &  $ 1.1218$ & $ 5.269\textsc{e}{-2} $ & $ 7.325\textsc{e}{-1} $ & $ \mathbf{5.243\textsc{e}{-2}} $ & $ \mathbf{7.1037\textsc{e}{-1}}$ & $ 1.8361 $ \\
        ACC &  $ 12.2757 $ & $ 1.337\textsc{e}{-2} $ & $ 110.913 $ & $ \mathbf{1.317\textsc{e}{-2}} $ & $ \mathbf{4.7796} $ & $ 6.8416\textsc{e}{+5} $ \\
        TORA & $ \mathbf{561.577591} $ & $ 6.656\textsc{e}{-1}$ & $ 983.758 $ & $ \mathbf{6.434\textsc{e}{-1}} $ & $ \times $ & $ 6.0325\textsc{e}+02 $ \\
        Unicycle & $ \mathbf{3940.6599} $ & $ 1.555\textsc{e}{-5} $ & $ 16348.3840 $ & $ \mathbf{9.7765\textsc{e}{-6}} $ & $ \times $ & $ \times $ \\
        \bottomrule
    \end{tabular}}
    \caption{\footnotesize Benchmark computation time (s) and set volumes, the black star (\ding{72}) denotes our approach. Computation times are listed for verified instances, and $ \times $ indicates an unverified instance. All available set volumes are shown. Best performance is highlighted in bold.}
    \label{tab:results}
\end{table*}

The Single Pendulum benchmark serves as a baseline across all tools. All methods complete the analysis in under two seconds, with CORA achieving the fastest runtime and OVERTVerify attaining the highest precision. This behavior is consistent with the design trade-offs of each method: CORA prioritizes computational efficiency at the expense of precision (producing a set approximately $35\times$ larger than the compact set obtained by OVERTVerify). OvertPoly balances these approaches, yielding a reachable set with volume approximately $0.6\%$ larger. 
The relative simplicity of this benchmark means that runtimes are dominated by bound computation and composition, which diminishes the advantages of OvertPoly’s approach, and results in slightly slower performance compared to OVERTVerify.
The design trade-offs are more pronounced in the ACC benchmark: both OvertPoly and CORA verify the property in under $15$ seconds, each more than $9\times$ faster than OVERTVerify. As expected, OVERTVerify yields the most precise results, while OvertPoly produces a solution that is approximately $2\%$ looser. The ACC specification depends only on a subset of the state variables, allowing CORA to successfully verify the benchmark despite the effects of state explosion.

As the benchmarks increase in complexity, CORA's lack of precision becomes a limiting factor, preventing it from verifying the TORA and Unicycle benchmarks. For the TORA benchmark, the specification depends on all state variables, and the reachable sets computed by CORA are overly conservative. In this setting, OvertPoly is approximately $1.8\times$ faster than CORA, while producing a solution that is approximately $4\%$ looser. 
For the Unicycle benchmark, state explosion causes CORA to fail (via a crash), whereas OvertPoly verifies the specification more than $4\times$ faster than OVERTVerify, producing a solution that is approximately $60\%$ looser.
These results highlight the relative strengths of each tool: CORA’s speed dominates on benchmarks it can successfully verify, whereas OVERTVerify achieves superior precision. OvertPoly trades a modest loss in precision for substantial improvements in computation time.
\subsection{Comparing combinatorial solvers}
The following experiments assess the strengths and limitations of our algorithm compared to OVERTVerify, the state-of-the-art combinatorial solver.
\subsubsection{Evaluating sensitivity to controller parameters}
To control for the peculiarities of the networks used in the ARCH-Competition benchmarks, we train additional networks for the TORA and Unicycle benchmarks. We use behavior cloning from an MPC expert and retain the architectures used in the original benchmark suite. To assess sensitivity to network size, we also train networks with the same architectures but with $0.5\times$ and $2\times$ as many neurons per hidden layer.
For the TORA benchmark, this yields networks with three hidden layers, each containing $k \in \{50,100,200\}$ neurons, and four outputs. For the Unicycle benchmark, this yields networks with one hidden layer containing $k \in \{250,500,1000\}$ neurons, and four outputs. All other properties of the corresponding benchmarks are unchanged. The results are shown in~\cref{tab:ab_ns}.
\begin{table*}[h!]
    \centering
    \resizebox{0.8\columnwidth}{!}{%
    \begin{tabular}{lcccc|cc}
        \toprule
        & \multicolumn{2}{c}{OvertPoly (\ding{72})}
        & \multicolumn{2}{c}{OVERTVerify}
        & \multicolumn{2}{c}{Improvement} \\
        \cmidrule(lr){2-3}
        \cmidrule(lr){4-5}
        \cmidrule(lr){6-7}
        & Time (s) & Vol. 
        & Time (s) & Vol.
        & Time ($\downarrow$) & Vol. ($\downarrow$) \\
        \midrule
        TORA (small)  & $1.97$ & $1.84\textsc{E-03}$ & $2.01$ & $1.76\textsc{E-03}$ & $1.015$ & $-1.043$ \\
        TORA (medium)  & $1.97$ & $1.85\textsc{E-03}$ & $2.16$ & $1.77\textsc{E-03}$ & $1.086$ & $-1.043$ \\
        TORA (large)  & $3.04$ & $1.84\textsc{E-03}$ & $2.91$ & $1.76\textsc{E-03}$ & $-1.048$ & $-1.043$ \\
        \midrule
        Unicycle (sma.)  & $1307.66$ & $4.89\textsc{E+01}$ & $9293.91$ & $3.78\textsc{E+01}$ & $7.107$ & $-1.292$ \\
        Unicycle (med.)  & $1165.37$  & $3.62\textsc{E+01}$ & $4782.088$ & $2.62\textsc{E+01}$ & $4.103$ & $-1.378$ \\
        Unicycle (larg.) & $2952.14$ & $5.43\textsc{E+01}$ & $13926.61$ & $2.32\textsc{E+01}$ & $4.717$ & $-2.337$ \\
        \bottomrule
    \end{tabular}}
    \setlength{\abovecaptionskip}{10pt}
     \caption{\footnotesize Benchmark computation times and set volumes using custom networks, the black star (\ding{72}) denotes our approach. The small network is $0.5\times$ larger than the standard number of neurons, medium is $1\times$, and large is $2\times$. For ease of comparison, we also report relative improvement in computation time and set volume}
    \label{tab:ab_ns}
\end{table*}
The trained networks appear to be better conditioned, as both benchmarks require significantly less time for verification. The previously observed pattern persists: OvertPoly achieves the fastest runtimes, while OVERTVerify yields the tightest bounds.
For the TORA benchmark, the problem becomes trivially solvable, suggesting that network evaluation dominated the solve time in the original benchmark. This also explains the limited separation between OvertPoly and OVERTVerify in that setting, since the baseline OvertPoly algorithm uses an identical network encoding. 
In contrast, the separation between OvertPoly and OVERTVerify is more pronounced on the Unicycle benchmark: OvertPoly is over $5\times$ faster on average, at the cost of solutions that are approximately $67\%$ looser. The dependence on network size is less straightforward; both algorithms exhibit degraded performance on the smallest and largest instances, while performing best on the medium-sized instance.

\subsubsection{Evaluating sensitivity to problem size}
Our approach to nonlinear composition raises concerns regarding both precision and scalability in higher-dimensional settings. To better characterize these limitations, we evaluate the method on a larger benchmark involving higher-dimensional multiplication. Specifically, we consider the \emph{Attitude Control} benchmark from the ARCH competition~\citep{lopez2023arch}, which involves compositions in a six-dimensional state space.

The objective is to verify the safety of a simplified model of an autonomous aircraft. The aircraft is modeled as a rigid body with six state variables: three orientation variables $(\psi_1,\psi_2,\psi_3)$, represented using Rodrigues parameters, and three angular velocity variables $(\omega_1,\omega_2,\omega_3)$. The original benchmark employs sigmoid activation functions, which are not supported by our current implementation. Following the approach described in the previous section, we instead train a ReLU network with the same architecture as in the original benchmark, but with twice the number of neurons per hidden layer.

Symbolic reachability inherently trades tractability for precision, making it a suitable setting for evaluating the scalability of combinatorial algorithms. \Cref{fig:scaling} highlights the scaling behavior of OvertPoly and \textsc{OVERTVerify} under symbolic reachability with $k$ symbolic steps. We track cumulative solve time and reachable set volume for $k \in \{1,\dots,20\}$. To emphasize relative differences between the tools, we use a logarithmic scale on the $y$-axis.
\begin{figure}[htpb]
    \resizebox{0.9\columnwidth}{!}{
    \input{figures/scaling_improved}}
    \caption{\footnotesize Symbolic reachability comparison between OverPoly and OVERTVerify using the modified Attitude benchmark. The $x$ axis denotes symbolic depth, while the $y$-axis denotes log (in base 10) cumulative computation time/volume}
    \label{fig:scaling}
\end{figure}

The results exhibit a consistent pattern: \textsc{OVERTVerify} achieves higher precision, while OvertPoly attains lower runtimes. The separation between the two methods becomes more pronounced as the symbolic depth increases. At $k = 20$, \textsc{OVERTVerify} produces reachable sets that are several orders of magnitude tighter than those obtained by OvertPoly. 
However, OvertPoly is consistently faster for $k \geq 4$, and this performance gap widens with increasing $k$. In practice, combinatorial solvers are often subject to fixed timeouts; \textsc{OVERTVerify} exceeds a three-hour timeout at $k = 3$, whereas OvertPoly remains below this threshold even at $k = 20$. These results indicate that \textsc{OVERTVerify} is likely to time out under conventional evaluation settings, while OvertPoly returns results more quickly, albeit with degraded precision.

\subsection{Ablations and Optimizations}
In this section, we evaluate the impact of our optimizations. The components included in the baseline algorithm (dependency graphs and McCormick envelopes) are treated as ablations, while the remaining components are evaluated as optimizations.
\subsubsection{Ablation: Dependency graph structure}
To evaluate the effect of the dependency graph structure, we compare our method against a variant that does not use dependency graphs (referred to as the \emph{flat} implementation). We conduct this evaluation on the Pendulum and ACC benchmarks. The results are reported in~\cref{tab:flat_vs_graph}.
\begin{table}[h]
\centering
\begin{tabular}{llr}
\toprule
Benchmark & Formulation & Time (s) \\
\midrule
Pendulum  & Flat  & $2.63$ \\
          & Graph & $1.07$ \\
\midrule
ACC       & Flat  & $186.13$ \\
          & Graph & $12.19$ \\
\bottomrule
\end{tabular}
\caption{\footnotesize Flat vs.\ Graph formulation runtime comparison}
\label{tab:flat_vs_graph}
\end{table}
The results are somewhat surprising: the flat encoding is $2\times$ slower on the Pendulum benchmark and more than $10\times$ slower on the ACC benchmark. This effect is more pronounced for ACC because, although the system is six-dimensional, each update function depends on only $2$--$3$ state variables. The flat implementation does not exploit this sparsity in the dependency structure, resulting in a significantly less efficient encoding.
\subsubsection{Ablation: Tightening nonlinear composition}
Next, we evaluate the effect of the bilinear tightening procedure described in~\cref{ssec:mccormick}. We compare our method against a variant that does not use the tightening procedure (which we refer to as the interval arithmetic implementation). We conduct this evaluation on the Unicycle benchmark with $\horizon = 5$. The results are reported in~\cref{tab:mccormick_ablation}. 
\begin{table}[h]
\centering
\begin{tabular}{lrr}
\toprule
Configuration & Time (s) & Volume \\
\midrule
McCormick    & $23.81$ & $9.76\textsc{E-07}$ \\
Interval Arith & $19.75$ & $2.22\textsc{E-06}$ \\
\bottomrule
\end{tabular}
\caption{\footnotesize McCormick envelope vs.\ interval arithmetic in nonlinear composition on the unicycle benchmark}
\label{tab:mccormick_ablation}
\end{table}
The tightening procedure yields a set volume that is $56\%$ smaller, at the cost of a modest overhead. We expect this effect to increase with $\horizon$, as errors accumulate across time steps.
\subsubsection{Optimization: Tightening pre-activation bounds}
We evaluate the potential benefits of tightening pre-activation bounds using the CROWN solver~\cite{zhang2018efficient}. To this end, we introduce a variant that replaces the MaxSens solver used in the baseline with CROWN. We conduct this evaluation on the TORA and Unicycle benchmarks (with $\horizon = 30$ for Unicycle), and report the results in~\cref{tab:crown_vs_maxsens}.
\begin{table}[h]
\centering
\begin{tabular}{llrrr}
\toprule
Benchmark & Bounds & Time (s) & Volume & ReLUs\\
\midrule
TORA      & MaxSens & $598.6$  & $6.656\textsc{E-01}$   & 126 \\
          & CROWN   & $519.6$  & $6.656\textsc{E-01}$   &  53 \\
\midrule
Unicycle  & MaxSens & $1941.3$ & $2.629\textsc{E-05}$ &   5 \\
          & CROWN   & $1948.6$ & $2.629\textsc{E-05}$ &   5 \\
\bottomrule
\end{tabular}
\caption{\footnotesize CC encoding: MaxSens vs.\ CROWN back-substitution. The ReLUs column denotes the number of active ReLUs in the big-M encoding}
\label{tab:crown_vs_maxsens}
\end{table}
Tightening pre-activation bounds yields moderate improvements in computation time for the TORA benchmark, reducing runtime by $13\%$. In contrast, tighter bounds have negligible impact on the Unicycle benchmark. This behavior is explained by the number of unstable ReLUs eliminated by CROWN: it removes $73$ unstable ReLUs in the TORA benchmark, but does not eliminate any in the Unicycle benchmark. This is consistent with expectations, as the Unicycle controller is relatively shallow. Overall, these results suggest that this optimization is most effective for deeper networks.

\subsubsection{Optimization: Improving LP relaxations for ReLU networks}
We evaluate the potential benefits of including the \emph{strengthening} constraints described in~\cref{ssec:anderson}. 
As we mentioned in that section, the procedure introduces an exponential number of constraints. We lexicographically add the first $N=100$ constraints obtained by the procedure, and report the results in~\cref{tab:anderson_cuts}
\begin{table}[h]
\centering
\begin{tabular}{llr}
\toprule
Benchmark & Condition & Time (s) \\
\midrule
TORA     & Baseline + CROWN, no cuts          & $519.6$ \\
                    & Baseline + CROWN, $N=100$ cuts & $304.81$ \\
\midrule
Unicycle & Baseline          & $1948.6$ \\
                    & Baseline + CROWN, $N=100$ cuts & $2285.45$ \\
\bottomrule
\end{tabular}
\caption{\footnotesize Ablation comparing the effect of strengthening cuts to the baseline. Both approaches use CROWN to tighten pre-activation bounds. Volumes are not reported as they are identical to the ones obtained using CROWN in~\cref{tab:crown_vs_maxsens}}
\label{tab:anderson_cuts}
\end{table}
The effects of the cuts are mixed, as they improve runtimes on TORA by $70\%$, but perform worse on the Unicycle benchmark. These results suggest that the trade off between \emph{strengthening} cuts and constraint overhead deserves further study. We leave this study to future work.
\subsubsection{Optimization: Compact encoding for polyhedral enclosures}
Finally, we evaluate the potential benefits of the compact encoding for polyhedral enclosures discussed in~\cref{ssec:dcc}. We use a lexicographic cell-wise Gray code as our injective function, ensuring that adjacent grid cells only differ in one bit. We present the results in~\cref{tab:dcc_vs_cc} 
\begin{table}[h]
\centering
\begin{tabular}{llr}
\toprule
Benchmark & Condition & Time (s) \\
\midrule
TORA     & Baseline+CROWN         & $519.6$ \\
                    &Baseline+CROWN, dcc & $354.37$ \\
\midrule
Unicycle & Baseline+CROWN          & $1948.6$ \\
                    &Baseline+CROWN, dcc& $1634.31$ \\
\bottomrule
\end{tabular}
\caption{\footnotesize Logarithmic disaggregated convex combination encoding vs aggregated convex combination encoding. Both approaches use CROWN to tighten pre-activation bounds. Volumes are not reported as they are identical to the ones obtained using CROWN in~\cref{tab:crown_vs_maxsens}.
}
\label{tab:dcc_vs_cc}
\end{table}
The compact encoding yield moderate improvements in computation time for both the TORA and Unicycle benchmarks, improving TORA by $46\%$ and Unicycle by $19\%$. The results are not dramatic, suggesting the number of binary variables is not the primary driver of solve times. Nonetheless, these results suggest that this optimization is effective across the board. 

%% file: sections/benchmarks.tex
\subsubsection{Single Pendulum}
We would like to verify that the angle of an inverted pendulum remains within a specified range during a specified time interval. Formally, we define the neural feedback system $\pend$, where $n^\pend = 2$, $I^\pend = [1.0, 1.2] \times [0.0, 0.2]$, $F^\pend(\vec{x}) = (\vec{x}_2, c_1\vec{x}_1 - c_2\vec{x}_2)$,$E^\pend = \{0\} \times \{0\}$, $u^\pend$ is computed by a neural network with two hidden layers (each with 25 neurons), and two outputs, the first of which is set to the constant zero value, $\delta^\pend = 0.05$, $T^\pend = 20$, and $G^\pend = \emptyset$. Let $R$ be the set defined by $[0,1] \times [-\infty, \infty]$, then $A^\pend$ is the sequence of sets defined by 
\[
    A^\pend(t) = 
    \begin{cases}
        \emptyset, &\text{if} \: t < 10 \\
        R^c &\text{if} \: 10 \leq t \leq 20
    \end{cases}
\]

We would like to verify that the avoid property holds for $\pend$.

\subsubsection{ACC}
In this benchmark, we would like to verify that a neural network controlled vehicle tracks a set velocity while maintaining a safe distance from a second vehicle.
Formally, we define the neural feedback system $\acc$, where $n^\acc = 6$, $I^\acc = [90,110] \times [32,32.2]\times\{ 0\}\times[10,11]\times[30,30.2]\times\{ 0\}$, $F^\acc(\vec{x}) = (\vec{x}_2, \vec{x}_3,-2\vec{x}_3 - c_1\vec{x}_2^2,\vec{x}_5, \vec{x}_6,$ $-2\vec{x}_6- c_1\vec{x}_5^2 )$,$E^\acc = \{0\} \times \{0\}\times \{0\}\times \{0\}$, $u^\acc$ is computed by a neural network with five hidden layers (each with 20 neurons), and six outputs, of which outputs 1,2,4, and 5 are set to the constant zero value, and output 3 is set to the constant -4 value, $\delta^\acc = 0.1$, $T^\acc = 50$, and $G^\acc = \emptyset$.  
Let $R$ be the set of states satisfying $\vec{x}_1 - \vec{x}_4 \geq 10 + 1.4 \cdot\vec{x}_5$
then $A^\acc$ is the sequence of sets defined by 
\[
    A^\acc(t) = 
    \begin{cases}
        R^c, &\text{if} \: 0 \leq t \leq 50 \\
        \emptyset &\text{if} \: t > 50 
    \end{cases}
\]
We would like to verify that the avoid property holds for $\acc$.
\subsubsection{TORA}
We would like to verify that the states of an actuated cart remain within a safe region during a specified time interval. Formally, we define the neural feedback system $\tora$, where $n^\tora = 4$, $I^\tora = [0.6, 0.7] \times [-0.7, -0.6] \times [-0.4,-0.3] \times [0.5, 0.6] $, $F^\tora(\vec{x}) = (\vec{x}_2, -\vec{x}_1+0.1\sin(\vec{x}_3), \vec{x}_4, 0)$,$E^\tora = \{0\} \times \{0\} \times \{0\}\times \{0\}$, $u^\tora$ is computed by a neural network with three hidden layers (each with 100 neurons), and four outputs, the first three of which is set to the constant zero value, $\delta^\tora = 0.1$, $T^\tora = 20$, and $G^\tora = \emptyset$. Let $R$ be the set defined by $[-2,2] \times [-2,2] \times [-2,2] \times [-2,2]$, then $A^\pend$ is the sequence of sets defined by 
\[
    A^\pend(t) = 
    \begin{cases}
        R^c, &\text{if} \: 0 \leq t \leq 20 \\
        \emptyset &\text{if} \: t > 20 
    \end{cases}
\]

We would like to verify that the avoid property holds for $\tora$.
\subsubsection{Unicycle Car Model}
See running example.

%% file: figures/scaling_improved.tex
% ============================================================
% Paste this into your document. Requires in preamble:
%   \usepackage{pgfplots}
%   \usepackage{subcaption}
%   \pgfplotsset{compat=1.18}
% ============================================================
\centering

\begin{subfigure}[b]{0.48\textwidth}
\centering
\begin{tikzpicture}
\begin{axis}[
    width=\textwidth,
    height=6.5cm,
    xlabel={Problem depth},
    ylabel={Log. cume. comp. time},
    legend style={at={(0.95,0.15)}, anchor=east, font=\footnotesize},
    grid=major,
    grid style={dashed, gray!30},
    tick label style={font=\footnotesize},
    label style={font=\small},
    xmin=1, xmax=21,
    ymin=2, ymax=6.5,
]

\addplot[color=blue, mark=*, mark size=1.2pt, thick] coordinates {
    (2,  2.656763036)
    (3,  3.844394845)
    (4,  4.436790314)
    (5,  4.731895491)
    (6,  4.982063496)
    (7,  5.144291921)
    (8,  5.262218469)
    (9,  5.354924265)
    (10, 5.431331854)
    (11, 5.496337488)
    (12, 5.552905029)
    (13, 5.602990075)
    (14, 5.647921535)
    (15, 5.688670399)
    (16, 5.725944834)
    (17, 5.760296354)
    (18, 5.792153375)
    (19, 5.821855506)
    (20, 5.849673630)
};
\addlegendentry{OVERTVerify}

\addplot[color=red, mark=triangle*, mark size=1.2pt, thick, dashed] coordinates {
    (2,  3.656742053)
    (3,  4.080980424)
    (4,  4.293654313)
    (5,  4.437450188)
    (6,  4.546753152)
    (7,  4.635132428)
    (8,  4.709583909)
    (9,  4.773933780)
    (10, 4.830582198)
    (11, 4.881368801)
    (12, 4.927590704)
    (13, 4.969818808)
    (14, 5.008917887)
    (15, 5.045209180)
    (16, 5.079204206)
    (17, 5.111095836)
    (18, 5.141407519)
    (19, 5.170101295)
    (20, 5.197429062)
};
\addlegendentry{OvertPoly}

\end{axis}
\end{tikzpicture}
\caption{\footnotesize Cumulative set computation time as a function of symbolic depth}
\label{fig:scaling_time}
\end{subfigure}
\hfill
\begin{subfigure}[b]{0.48\textwidth}
\centering
\begin{tikzpicture}
\begin{axis}[
    width=\textwidth,
    height=6.5cm,
    xlabel={Problem depth},
    ylabel={Log. cume. set volume},
    legend style={at={(0.05,0.85)}, anchor=west, font=\footnotesize},
    grid=major,
    grid style={dashed, gray!30},
    tick label style={font=\footnotesize},
    label style={font=\small},
    xmin=1, xmax=21,
]

\addplot[color=blue, mark=*, mark size=1.2pt, thick] coordinates {
    (2,  -11.04426057)
    (3,  -10.55081479)
    (4,  -10.11895966)
    (5,  -9.679998825)
    (6,  -9.277175429)
    (7,  -8.916970351)
    (8,  -8.598714615)
    (9,  -8.305389835)
    (10, -7.994600787)
    (11, -7.636046489)
    (12, -7.285809091)
    (13, -6.879397894)
    (14, -6.464890129)
    (15, -6.04830531)
    (16, -5.58311661)
    (17, -5.06683497)
    (18, -4.428857036)
    (19, -3.617003171)
    (20, -2.662260187)
};
\addlegendentry{OVERTVerify}

\addplot[color=red, mark=triangle*, mark size=1.2pt, thick, dashed] coordinates {
    (2,  -10.55299533)
    (3,  -9.856997809)
    (4,  -9.268989639)
    (5,  -8.724565803)
    (6,  -8.205576504)
    (7,  -7.703993171)
    (8,  -7.215562559)
    (9,  -6.737209808)
    (10, -6.261939426)
    (11, -5.775522433)
    (12, -5.270298439)
    (13, -4.747158822)
    (14, -4.202832867)
    (15, -3.633547193)
    (16, -3.018596869)
    (17, -2.319847206)
    (18, -1.489782102)
    (19, -0.4516736853)
    (20,  3.558722472)
};
\addlegendentry{OvertPoly}

\end{axis}
\end{tikzpicture}
\caption{\footnotesize Cumulative reach.set volume as a function of symbolic depth}
\label{fig:sub2}
\end{subfigure}

%% file: sections/endSections.tex
In this work, we showed that using bounding sets and polyhedral enclosures provides a scalable abstraction for verifying reach-avoid properties of nonlinear neural feedback systems. Polyhedral enclosures (along with neural network controllers) are encoded as mixed integer linear programs, enabling forward reachability analysis of discrete time systems. These techniques are integrated in the OvertPoly algorithm, which shows a significant improvement in both computation time and precision when compared to existing tools. The improved scalability of precise verification tools enables the verification of more realistic neural feedback systems, which is a promising step toward safer autonomous transportation systems.
 
\subsection{Limitations}
\modded{Added this section}
Our approach is limited to systems whose transition functions can be represented as extended rational nonlinear functions. We are also restricted to systems with ReLU activations.
\subsection{Future work}
\modded{Modified this to reflect the work already done}
We would like to extend this abstraction to a larger class of nonlinear transition functions. We would also like to integrate our abstraction algorithm with more capable representations for the neural network controller.
\acks{We would like to acknowledge David Cole for his invaluable feedback on the dependency graph structure, as well as David Dill and other members of the Centaur lab for their feedback on the paper.
This material is based upon work supported by the National Science
Foundation Graduate Research Fellowship Program under Grant No. (2146755).
Any opinions, findings, and conclusions or recommendations expressed in this
material are those of the author(s) and do not necessarily reflect the views
of the National Science Foundation.  Additional support was provided by NSF Grant No. 2211505 and by the Stanford Center for Automated Reasoning.}
\modded{Added acknowledgements as well}